\documentclass[peerreview,onecolumn,12pt,draftclsnofoot]{IEEEtran}
\usepackage{amsmath,amsthm}
\usepackage{amssymb,graphicx,url}
\usepackage{cite}
\usepackage{mathtools}
\usepackage{bm}
\usepackage{float}
\usepackage[capitalize]{cleveref}
\DeclareTextSymbolDefault{\DH}{T1}
\DeclareTextSymbolDefault{\DJ}{T1}
\usepackage{enumitem}
\setlist[enumerate]{leftmargin=.5in}
\setlist[itemize]{leftmargin=.5in}
\DeclarePairedDelimiter\bra{\langle}{\rvert}
\DeclarePairedDelimiter\ket{\lvert}{\rangle}
\newcommand{\kett}[1]{\bm{#1}}
\newcommand{\braa}[1]{\adj{\bm{#1}}}

\newcommand{\norm}[1]{\left\lVert {#1} \right\rVert}
\newcommand{\bbR}{\mathbb{R}}
\newcommand{\bbC}{\mathbb{C}}
\newcommand{\cP}{\mathcal{P}}
\newcommand{\cX}{\bm{\mathcal{X}}}
\newcommand{\cH}{\mathcal{H}}

\newcommand{\tar}{\sigma_1^2(X)+\sigma_2^2(X)}
\newcommand{\half}{\frac{1}{2}}

\newcommand{\iu}{\mathrm{i}\mkern1mu}
\newcommand{\trace}[1]{\mathrm{Tr}(#1)}
\usepackage{dsfont}
\newcommand{\I}{\mathds{I}}

\newcommand{\adj}[1]{#1^{\dagger}}
\newcommand{\trans}[1]{#1^{\intercal}}
\newcommand{\Real}[1]{\mathrm{Re}(#1)}
\newcommand{\Ima}[1]{\mathrm{Im}(#1)}
\newcommand{\etal}{{\textit{et al.}}}
\usepackage{xspace}
\newcommand{\setS}{\mathcal{S}}%

\newtheorem{definition}{Definition}
\newtheorem{proposition}[definition]{Proposition}
\newtheorem{lemma}[definition]{Lemma}

\newtheorem{theorem}[definition]{Theorem}
\newtheorem{corollary}[definition]{Corollary}
\newtheorem{conjecture}[definition]{Conjecture}

\begin{document}

\title{A Matrix Inequality Related to the Entanglement Distillation Problem}
\author{Lilong Qian, Lin Chen, Delin Chu, Yi Shen
\thanks{This work was
  supported by the NUS Research Grant R-146-000-281-114, the NNSF of China (Grant No. 11871089), and the Fundamental
  Research Funds for the Central Universities (Grant Nos.  KG12080401 and ZG216S1902).
  {\sl (Corresponding author: Lin Chen)}}%
\thanks{L. Qian and D. Chu are with the Department of Mathematics, National University of Singapore, Singapore,
  119076, Singapore (e-mail: qian.lilong@u.nus.edu; matchudl@nus.edu.sg).}%
\thanks{L. Chen is with the School of
  Mathematics and Systems Science, Beihang University, Beijing, 100191, China, and also with the International
  Research Institute for Multidisciplinary Science, Beihang University, Beijing 100191, China
  (e-mail: linchen@buaa.edu.cn). }%
\thanks{Y. Shen is with the School of Mathematics and Systems Science,  Beihang University, Beijing, 100191, China
  (e-mail: yishen@buaa.edu.cn).}%
}

 \maketitle

\begin{abstract}
  The pure entangled state is of vital importance in the field of quantum information. The process of asymptotically
  extracting pure entangled states from many copies of mixed states via local operations and classical communication is
  called entanglement distillation. The entanglement distillability problem, which is a long-standing open problem, asks
  whether such process exists. The 2-copy undistillability of $4\times4$
  undistillable Werner states  has been reduced to the validness of the a matrix inequality, that is,
  the sum of the squares of the largest two singular values of matrix $A\otimes \I + \I \otimes B$ does not exceed $(3d-4)/d^2$
  with
  $A,B$ traceless $d\times d$ matrices and $\norm{A}_F^2+\norm{B}_F^2=1/d$ when $d=4$.
  The latest progress, made by {\L}.~Pankowski~\etal~[IEEE Trans. Inform. Theory, 56, 4085
    (2010)], shows that this conjecture holds when both matrices $A$ and $B$ are normal. In this paper, we prove that the conjecture
  holds when one of matrices $A$ and $B$ is normal and the other one is arbitrary. Our work makes solid progress towards
  this conjecture and thus the distillability problem.
\end{abstract}
\IEEEpubid{}
\IEEEpeerreviewmaketitle

\section{Introduction}
\IEEEPARstart{T}{he} entanglement is a fundamental resource in the field of quantum information \cite{Nielsen}.
It is of great importance
for superdense coding~\cite{Harrow_2004}, teleportation~\cite{Bouwmeester_1997}, quantum computing~\cite{Jozsa_2003},
and cryptography~\cite{Ekert_1991,Gisin_2002}. Although some mixed states can be used directly~\cite{Murao_2001}, pure
entangled states play an essential role in most quantum-information tasks \cite{Bennett_1993,Briegel_2001}. Obviously,
there is no pure state in nature due to the inevitable decoherence between the state and environment. Therefore,
asymptotically converting initially bipartite entangled mixed states into bipartite pure entangled states under local
operations and classical communications (LOCC) is a key step in quantum information processing. The above-mentioned
conversion is also known as entanglement distillation. It is natural to ask whether all entangled states can be
distilled. It is also famously known as the distillability problem. A bound entangled state of a bipartite system is one
which cannot be distilled. The phenomenon of bound entanglement lies in the centre of entanglement theory. Therefore,
the distillability problem has been a main open problem in entanglement theory for a long time.

In order to describe the distillability problem explicitly, we first introduce the basic mathematical preliminaries for
quantum information theory. Mathematically, any quantum state can be described by a positive semidefinite Hermitian
matrix of trace one \footnote{The condition of trace one is required for explaining quantum states by the hypothesis of
  quantum physics. For conveniently treating mathematical problems in quantum information such as the distillability
  problem, we may omit the condition unless stated otherwise.} , namely the {\em density matrix} or {\em density
  operator}~\cite{Nielsen}.  If the rank of density matrix is one, then we refer to it as a {\em pure state}. Otherwise,
we call it a {\em mixed state}.  For the composite system, the $N$-partite Hilbert space is described by the tensor
space
\begin{equation}
  \cH_1\otimes \cdots\otimes \cH_N
  = \mathrm{span}\{ \ket{x_1,\ldots,x_N},\ket{x_i}\in\cH_i\},
\end{equation}
where each $\cH_i$ is a Hilbert space corresponding to the $i$-th system.  The $N$-partite quantum state $\rho$ is a
positive semidefinite Hermitian operator acting on the space $\cH_1\otimes \nobreak\cH_2\otimes\cdots\otimes
\cH_N$. Specifically, $\rho$ is said to be {\em separable} if it admits the following decomposition:
\begin{equation}
  \rho = \sum_{i=1}^L p_i\ket{x_{i}^{(1)},x_{i}^{(2)},\ldots,x_{i}^{(N)}}\bra{x_{i}^{(1)},x_{i}^{(2)},\ldots,x_{i}^{(N)}},
\end{equation}
where $p_i> 0$ and $L>0$ Otherwise, the state is said to be {\em entangled}. It is NP-hard to to check whether a given
state is entangled despite of progress in the past decades~\cite{Hor09,Qian_sppt,Qian_CS,Qian_CS2}.

The {\em maximally entangled state} is a bipartite pure state which can be brought by a local change of basis to the
state
\begin{equation}
  \ket{\Phi_{\max}}
  =\frac{1}{\sqrt{d}}\sum_{i=1}^d\ket{i,i}\in\bbC^d\otimes \bbC^d.
\end{equation}
 Here the local change of basis
corresponds to the invertible local operator $A_1\otimes A_2$.  The maximally entangled state can be used for
transmitting qubits by means of teleportation~\cite{Ekert_1991}.  However, there do not exist pure and maximally
entangled states naturally.  Therefore, the idea of asymptotically or explicitly converting mixed entangled states into
maximally entangled states by using LOCC has been introduced by
Bennett~\etal~\cite{Bennett_1996,Bennett_1996_concentrating}, Deutsch~\etal~\cite{Deutsch_1996_privacy}, and
Gisin~\cite{Gisin_1996_hiden}. It has been a central topic in quantum information theory so
far~\cite{Bombin_2006,Fang_2019}.  It is known that the {\em entanglement distillation} is equivalent to extracting
maximally entangled states from mixed entangled states.  Now we can present the formal description of distillability as
follows:
\begin{definition}
  $\rho$ is said to be $K$-distillable or $K$-copy distillable if 
  $K$ copies of $\rho$ can be transformed arbitrarily close to $\ket{\Phi_{\max}}$
  via LOOC, that is
  \begin{equation}
    \underbrace{\rho\otimes \rho\otimes \cdots\otimes \rho}_{K \text{copies}}\xrightarrow{LOOC} \ket{\Phi_{\max}}\bra{\Phi_{\max}}.
  \end{equation}
  Otherwise, it is said to be $K$-undistillable or $K$-copy undistillable.  $\rho$ is said to be distillable if it is
  $K$-distillable for some number $K$. Otherwise, it is said to be undistillable or bound entangled.
\end{definition}

Based on the above definition, in order to determine whether a given state is distillable, we need to consider all
possible kinds of LOCC. This is a hopeless task in realistic situation. Fortunately, Horodecki
\etal~\cite{Horodecki_1998_Is} has constructed an equivalent definition of distillability. It turns the distillation
problem to a precisely stated mathematical result.
\begin{theorem}\label{thm:dis}
 Given a bipartite state $\rho$ acting on  $\cH_1\otimes\cH_2$, it  is called  $K$-distillable if and only if there exists a
  Schmidt-rank-two bipartite pure state $\ket{\psi}\in(\cH_1^{\otimes n})\otimes \nobreak(\cH_2^{\otimes n})$ such that
  \begin{equation}
    \langle \bm \psi, ({\rho^{\otimes n}})^\Gamma \bm \psi\rangle <0.
  \end{equation}
\end{theorem}

To understand the theorem, we explain the Schmidt rank and notation $\Gamma$, respectively.

First, for any pure
state $\ket{\phi}\in\cH_1\otimes \cH_2$, there  exist orthonormal sets
$\{\ket{ u_1},\ket{ u_2},\ldots,\ket{ u_R}\}\subset \cH_1$ and
$\{\ket{ v_1},\ket{ v_2},\ldots,\ket{ v_R}\}\subset \cH_2$ such that
\begin{equation}
  \ket\phi  = \sum_{i=1}^Rp_i \ket{ u_i,v_i},\quad p_i>0.
\end{equation}
This decomposition is called the {\em Schmidt decomposition} of $\ket\phi$ and $R$ is called the {\em Schmidt rank} of
$\ket{\phi}$. Next we assume that $\dim\cH_1=M$ and $\dim\cH_2=N$. Denote by $E_{ij}$ the $M\times M$ matrix whose
elements are all zero, except that the $(i,j)$-th entry is one. Hence, any given state $\rho$ can be represented by
\( \rho = \sum_{i,j} E_{ij}\otimes \rho_{ij}\), where $\rho_{ij}$ are operators acting on the Hilbert space $\cH_2$.
The partial transpose of $\rho$ is defined by \( \rho^\Gamma= \sum_{ij}E_{ji}\otimes \rho_{ij}\).

From now on, we use Theorem~\ref{thm:dis} as the definition of distillability.  The state $\rho$ is called PPT (Positive
Partial Transposed) if $\rho^\Gamma$ is positive semidefinite, i.e., $\rho^\Gamma \geqslant 0$.  Otherwise, $\rho$ is
called NPT (Non-positive Partial Transposed).  According to the Theorem~\ref{thm:dis}, all the PPT entangled states are
undistillable, i.e. the bound entangled states. The existence of bound entanglement is striking, since it implies
irreversibility: to create them by LOCC one needs pure entanglement, but no pure entanglement can be obtained back from
them \cite{pphh2010}.  The famous Horodecki-Peres criterion \cite{hhh96,peres1996} tells that all NPT states are
entangled.  The question ``Are all the NPT states are distillable'' remains open due to the importance of pure states
and bound entanglement. The question is equivalent to distillability problem. Despite many efforts devoted to the
distillability problem over the past decades \cite{dss00,dcl00,pphh2010,vd06,br03,cc08,hhh97,rains1999,cd11jpa,cd16pra},
it is still an open problem.  However, some partial solutions have been found. For example, entangled states of ranks
$2$ and $3$~\cite{Horodecki_2003,cc08,hhh97}, $2\otimes N$ NPT states~\cite{dcl00}, and $M\otimes N$ NPT states of rank
at most $\max\{M,N\}$~\cite{Horodecki_2003,cd11jpa} have been proven to be distillable. In Ref.~\cite{cd16pra}, the
authors proved that rank-four and two-qutrit NPT bipartite states are distillable.  Moreover, it has been shown in
Ref.~\cite{dss00} that all NPT bipartite states can be locally converted into the NPT Werner states. Hence, it suffices
to consider the distillability problem of Werner states on $\bbC^d\otimes \bbC^d$. The state is defined as
\begin{equation}\rho_W(\alpha)=\frac{\I + \alpha \sum_{i,j=1}^d E_{ij}\otimes
    E_{ji}}{d^2+\alpha d},\end{equation} where $\alpha \in [-1,1]$.  The following results divide the Werner states into three
different cases.
\begin{proposition}[D. P. DiVincenzo \etal~\cite{hh1999, Rungta, dss00,dcl00}]
  The Werner states $\rho_W(\alpha)$ are
  \begin{enumerate}[label=(\alph*)]
    \item separable for $-1\leqslant \alpha\leqslant \frac{1}{d}$;
    \item NPT and one-distillable for $\half < \alpha\leqslant 1$;
    \item NPT and one-undistillable for $\frac{1}{d}<\alpha\leqslant \frac{1}{2}$.
  \end{enumerate}
\end{proposition}

\noindent Hence to investigate the distillability problem, it suffices to consider case (c). We need investigate whether
$\rho_w(\alpha)$ is $K$-distillable for some $K\geqslant 2$. It is proved that if $\rho_W(\half)$ is
$K$-undistillable, then the states $\rho_W(\alpha),\;\frac{1}{d}<\alpha\leqslant \half,$ are $K$-undistillable
\cite{dss00}.  So we only consider the distillability of $\rho_w(\half)$. Nevertheless it is widely believed that
$\rho_w(\half)$ is not distillable~\cite{dss00,dcl00,Clarisse06}. Some equivalent formulations and evidence for the
validity of the distillability problem are provided in Ref.~\cite{okovi2016}.

In this paper we investigate the $2$-distillability of Werner states in $\bbC^4\otimes \bbC^4$. It is proved in
Ref.~\cite{pphh2010} that these states are two-undistillable if and only if the following conjecture holds.

\begin{conjecture}%
  \label{cj:main-4}
   Let $\sigma_i(X)$ be the $i$-th largest singular value of $X$. Then
  \begin{equation}
    \begin{aligned}
      \label{eq3}
    \sup_{X\in\cX_4}  \tar \leqslant \frac{1}{2}.
    \end{aligned}
  \end{equation}
\end{conjecture}
Here   $\cX_d$ denotes the set of matrices $ X=A\otimes \I+\I\otimes\nobreak B$ which satisfies the conditions
$\trace A=\trace B=0$, and  $\norm{A}_F^2+\norm{B}_F^2=\frac{1}{d}$.

One can show that Conjecture \ref{cj:main-4} is a special case of the following general conjecture.

\begin{conjecture} \label{cj:high_var}
  \begin{equation}
    \label{complex}
    \sup_{X\in\cX_d}\sigma_1^2(X) + \sigma_2^2(X) \leqslant \frac{3d-4}{d^2},\quad d\geqslant 4.
  \end{equation}
\end{conjecture}

Existing results on Conjecture's \ref{cj:main-4} and \ref{cj:high_var} can be summarized in the following three
theorems.

\begin{theorem}[Y. Shen \& L. Chen \cite{ysdis2018}] \label{le:cj} The following five statements are equivalent:
  \begin{enumerate}
    \item Conjecture  \ref{cj:main-4} (also Conjecture \ref{cj:high_var}) holds.
    \item Conjecture  \ref{cj:main-4} (also Conjecture \ref{cj:high_var}) holds when $X$ is replaced by $X^{{}^\intercal}$, $\overline{X}$ or $\adj X$.
    \item Conjecture  \ref{cj:main-4} (also Conjecture \ref{cj:high_var}) holds when $X$ is replaced by any matrix locally unitarily similar to $X$.
    \item Conjecture  \ref{cj:main-4} (also Conjecture \ref{cj:high_var}) holds when $X$ is replaced by $\I\otimes  A+B\otimes \I$.
    \item Conjecture  \ref{cj:main-4} (also Conjecture \ref{cj:high_var}) holds when $X$ is replaced by $e^{\iu \theta}X$ for some $\theta\in[0,2\pi]$.
  \end{enumerate}
\end{theorem}

\bigskip

\begin{theorem} Conjecture \ref{cj:main-4} holds when either of the following two conditions holds.
  \begin{enumerate}
    \item Both  matrices $A$ and $B$ are normal or unitarily similar to
      \begin{equation}
        \begin{bmatrix}
          0&b_1&0&0\\
          b_2&0&0&0\\
          0&0&0&b_3\\
          0&0&b_4&0
        \end{bmatrix}.
      \end{equation}
    \item  One of matrices $A$ and $B$ is normal and the other one is unitarily similar to
      \begin{equation}
      \begin{bmatrix}
        0 & 0 & 0 & b_2 e^{\iu \theta} \\
        b_1 & 0 & 0 & 0 \\
        0 & b_2 & 0 & 0 \\
        0 & 0 & b_1 & 0
      \end{bmatrix}.
      \end{equation}
  \end{enumerate}
\end{theorem}

\bigskip

\begin{theorem}[\L{} Pankowski, M. Piani, M.  Horodecki, \& P. Horodecki~\cite{pphh2010}] \label{lem:normal} Conjecture
  \ref{cj:high_var} holds when both matrices $A$ and $B$ are normal.
\end{theorem}

We claim that Theorem \ref{lem:normal} is equivalent to the following corollary.

\begin{corollary}\label{lem:normal2}
  Let $A,B\in\bbC^{d\times d} (d\geqslant4)$, $\trace A=\trace B=0$ and both $A$ and $B$ are normal. Denote
  \begin{equation}
    X=A\otimes \I+\I\otimes B.
  \end{equation} Then
  \begin{equation} \sigma_1^2(X) + \sigma_2^2(X) \leqslant \frac{3d-4}{d}\left(\norm{A}_F^2+ \norm{B}_F^2\right).
  \end{equation}
\end{corollary}


Although extensively numerical tests have demonstrated the validness of Conjecture \ref{cj:main-4} and
\ref{cj:high_var} (see \cref{fig}), they have been open problems since the last progress was made in 2010 \cite{pphh2010}.  We
investigate Conjecture \ref{cj:high_var} in this paper.  Our main result is as follows, and will be proven in
Sections~\ref{sec3}.
\begin{figure}[!htb]
  \centering
  \includegraphics[width=\linewidth]{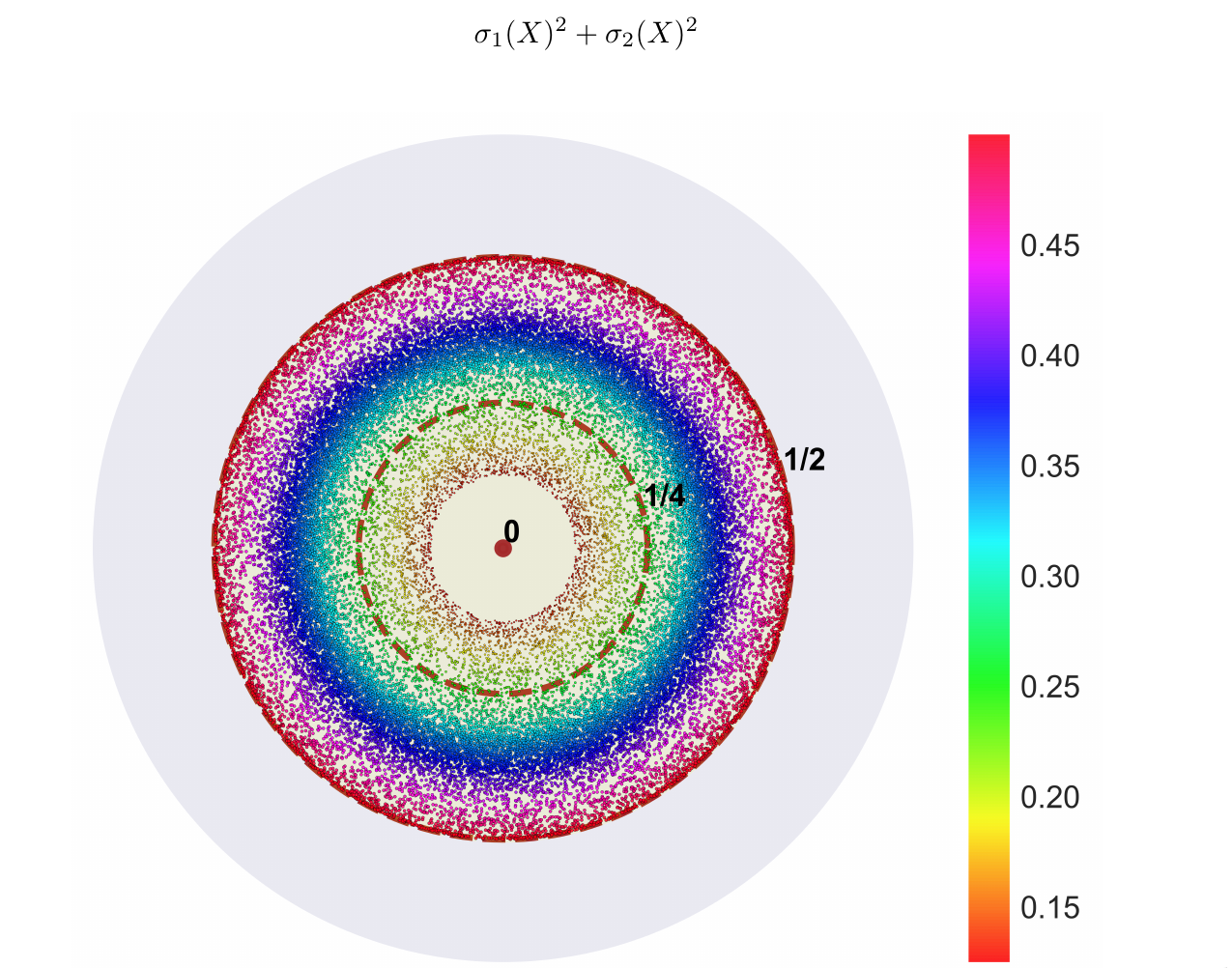}
  \caption{Numerical evidence of the validity of \cref{cj:main-4} where the samples $A, B$ are generated randomly (including normal and non-normal) with uniform distribution,
    and the distance
    from the origin (or colors) in the polar graph represents the value of $\tar$.}
  \label{fig}
\end{figure}
\begin{theorem}\label{main-result1}
  Conjecture \ref{cj:high_var} holds when one of matrices $A, B\in\bbC^{d\times d}$ is normal and the other one is
  arbitrary.
\end{theorem}

It is obvious that our work leads to significant progress on Conjecture \ref{cj:high_var} (so, also Conjecture
\ref{cj:main-4}), as well as towards the distillability problem. If we relax the condition on the square sum of the
norms of $A$ and $B$, Theorem \ref{main-result1} reduces to the following result.

\

\begin{corollary}\label{coro1}
  Let $A,B\in\bbC^{d\times d} (d\geqslant4)$, $\trace A=\trace B=0$ and one of $A$ and $B$ is normal and the other one
  is arbitrary. Denote
  \begin{equation} X=A\otimes \I+\I\otimes B, \end{equation} and let $\sigma_i(X)$ be the $i$-th largest singular value of $X$. Then
  \begin{equation} \sigma_1^2(X) + \sigma_2^2(X) \leqslant \frac{3d-4}{d}\left(\norm{A}_F^2+ \norm{B}_F^2\right).
  \end{equation}
\end{corollary}


Actually, the upper bound in Theorem \ref{main-result1} is attainable. For example, Let
\begin{equation} A= \mathrm{diag}(2(d-1)\beta,-2\beta,\ldots,-2\beta),\end{equation} and
\begin{equation}
  B= \mathrm{diag}( (d-2)\beta,(d-2)\beta,-2\beta,\ldots,-2\beta), \end{equation} where
\begin{equation}
  \beta = \left(d\sqrt{6d-8}\right)^{-1}.
\end{equation}
Then $\trace A=\trace B=0$, $\norm{A}_F^2+\norm{B}_F^2 =\frac{1}{d}$ and
\begin{equation}
  \sigma_1^2(X) + \sigma_2^2(X) = \frac{3d-4}{d^2}.
\end{equation}

Hence, we have the following results.
\begin{theorem}\label{main-result11}
  Denote by $\cX^{(1)}_d$ the subset of $\cX$ where one of $A$ and $B$ is normal. Then
  \begin{equation} \max_{X\in \cX^{(1)}_d} \sigma_1^2(X) + \sigma_2^2(X) = \frac{3d-4}{d^2}.
  \end{equation}
\end{theorem}

\begin{corollary}\label{main-result22}
  Let $\cX^{(2)}_d$ be the set of matrices $A\otimes \I+\I\otimes B$ with $A,B$ traceless and   one of $A$ and $B$  normal. Then
  \begin{equation} \max_{X\in \cX^{(2)}_d} \sigma_1^2(X) + \sigma_2^2(X)= \dfrac{3d-4}{d}\left(\norm{A}_F^2+ \norm{B}_F^2\right).
  \end{equation}
\end{corollary}

The rest of the paper is organized as follows.  In Section~\ref{sec:pre},  we will
introduce some  necessary mathematical notations and the
background of optimization theory. In particular, we provide two useful lemmas
which will be used in the proof of our main result repeatedly.
 In Section~\ref{sec3}, the proof of our main result is provided, that is, prove
the validness of \cref{cj:high_var} when  $A$ is normal.
Finally, some concluding remarks are given in Section~\ref{sec:concluding}.
\section{Preliminaries}
\label{sec:pre}
\subsection{Mathematical Notations}
Let us first introduce some mathematical notations. We refer to $\bbC^{d\times d}$, $\bbR^{d\times d}$, and
$\mathbb{H}^{d\times d}$ as the $d\times d$ complex matrices, real matrices, and Hermitian matrices, respectively.

To be clearly, hereafter in this paper, we use the capital letters to indicate the matrices, for example,
$A\in \bbC^{d\times d}$.  On the other hand, we use the bold lowercase letters to indicate the vectors, for example,
$\bm{a}\in \bbC^d$.  In particular, throughout this paper, we use $\kett{e}_i,i=1,\ldots,d$ to indicate the natural basis in
$\bbC^d(\text{ or } \bbR^d)$ , that is, $\kett{e}_i$ is the vector whose $i$-th entry is one and all the others are
zeros.

By convention, $a_i$ is denoted as the $i$-th entry of the vector $\kett{a}$ and $a_{ij}$ the $(i,j)$-th entry of the
matrix $A$.

For any matrix $A=(a_{ij})\in\bbC^{d\times d}$, we use $\norm{A}_F$ to denote the {\em Frobenius Norm}\,:
\begin{equation}
  \norm{A}_F = \sqrt{\sum_{ij=1}^d|a_{ij}|^2}.
\end{equation}
Another important operator for matrices is the {\em trace operator}\,:
\begin{equation}
  \trace{A} = \sum_{i=1}^da_{ii}.
\end{equation}
In addition, we have the relation:
\begin{equation}
  \trace{A\adj A} = \trace{\adj A A} = \norm{A}_F^2,
\end{equation}
where $\adj A$ indicates the conjugate transpose of $A$. Moreover, $A^\intercal$ indicates the transpose of the matrix
$A$.  A matrix $A$ with $M$ rows and $N$ columns is called an $M\times N$ matrix, while $M$ and $N$ are called its
dimensions.  By convention, we use $\mathrm{dim}(A)$ to indicate the dimension if $A$ is square, i.e.,
$\mathrm{dim}(A)=N$. Morevoer, we say a matric $A$ is {\em traceless} when $\trace{A}=0$.

For a vector $\kett a\in\bbC^{d}$, we use $\norm{\bm a}$ to denote the {\em Euclidean norm}, that is,
\begin{equation}
  \norm{\bm a} = \sqrt{\sum_{i=1}^d|a_i|^2}
  =\sqrt{\trace{\kett{a}\braa{a}}} =\sqrt{\langle \kett{a}, \kett{a} \rangle }.
\end{equation}


Throughout this paper, we use $\I$ to denote the {\em identity operator}. In particular, $\I_d$ indicates the identity
operator of order $d$. If no other specified, the order of $\I$ should match the operations in context.

Given two matrices $A=(a_{ij})$ and $B=(b_{ij})$ in $\bbC^{d\times d}$, the {\em tensor product} (Kronecker
product~\cite{Horn_1991}) is defined as
\begin{equation}
  A\otimes B =
  \begin{bmatrix}
    a_{11}B &\cdots & a_{1d}B\\
    \vdots & \ddots & \vdots \\
    a_{d1}B& \cdots & a_{dd}B
  \end{bmatrix}.
\end{equation}
For any two Hilbert spaces $\cH_1$ and $\cH_2$, the {\em tensor space} is defied by
\begin{equation} \cH_1\otimes \cH_2 = \{\ket u\otimes \ket v: \ket u\in \cH_1,\ket v\in\cH_2\}. 
\end{equation}

In addition, the {\em direct sum} $A\oplus B$ is defined as
\begin{equation}
  A\oplus B =
  \begin{bmatrix}
    A&0\\
    0&B
  \end{bmatrix}
\end{equation}
for any two given square matrices.

For any matrix $A$, denote by $\lambda_i(A)$ and $\sigma_i(A)$ the $i$-th largest eigenvalue and singular value of $A$,
respectively. For any two matrices $X$ and $Y$, we say $X$ is {\em unitarily similar} to $Y$, i.e., $X\sim Y$, when
there exists a unitary operator $U$ such that $X = UY\adj U $. In particular, in the composite system
$\cH_1\otimes \cH_2$, $X$ is said to be {\em locally unitarily similar} to $Y$ if there exist unitary operators $U$ and
$V$ acting on $\cH_1$ and $\cH_2$, respectively, such that $X = (U\otimes V)Y\adj{(U\otimes V)}$, where $X$ and $Y$ are
the operators acting on $\cH_1\otimes \cH_2$.

Let $\cP$ denotes the set of all matrices which is locally unitarily similar to
\begin{equation}
  \oplus_{i=1}^lB_i,
\end{equation}
where each $B_i$ is either an $1\times 1$ or $2\times 2$ matrix and $\sum_{i=1}^l\mathrm{dim}(B_i)=d$.

As for the singular values, there exists a well-known Weyl's inequality.
\begin{lemma}[Weyl's inequality~\cite{Weyl_1912}]
  Let $A$ and $B$ are two matrices in $\bbC^{d\times d}$. Then
  \begin{equation}
    \sigma_{i+j-1}(A+B)\leqslant \sigma_i(A)+\sigma_j(B),
  \end{equation}
  for all $1\leqslant i,j\leqslant d, i+j\leqslant d$.  In particular,
  \begin{equation}
    \sigma_1(A+B) \leqslant \sigma_1(A)+\sigma_1(B).
  \end{equation}
\end{lemma}


\subsection{Supporting Lemmas in Optimization Theory}

Our developments in the next section will be heavily relied on the optimization theory. Given a function $f$, the
maximization problem with linear and quadratic constraints can be formulated as:
\begin{equation}\label{eq:pre:max}
  \begin{array}{rl}
    \max & \quad f(\kett{x})\\
    \text{s.t.}& \quad
                 \begin{cases}
                   \langle \kett{c}_i,\kett{x}\rangle &=0, \quad i=1, \ldots, m, \\
                   \braa{x} W \kett{x} &=r,
                 \end{cases}
  \end{array}
\end{equation}
where $\kett{c}_i\in\bbR^N,\;\kett{x}\in\bbR^N$, $W\in\bbR^{N\times N}$, and $r\in\bbR$.  Here $f$ is usually called the
{\em objective function} which is desired to be maximized.

In particular, the quadratic constraints in this paper always appear as the following form:
\begin{equation}
  \sum_{i=1}^Nw_ix_i^2=r,\quad w_i>0,\quad r>0.
\end{equation}
That is, $W$ is a positive diagonal matrix.  Moreover, the objective function $f$ is always continuous differentiable
and degree-$2$ homogeneous.  Recall that, a function $f(\kett{x})$ with $\kett{x}\in\bbR^N$ is {\em degree-$k$
  homogeneous} $(k>0)$ if
\begin{equation}
  f(t\kett{x}) = t^k f(\kett{x}),\quad \forall\; t\in\bbR.
\end{equation}

In addition, we say that a vector $\kett{x}$ is a {\em feasible point } of the optimization problem~\eqref{eq:pre:max}
if it satisfies all the constraints of the optimization problem, i.e., the linear and quadratic constraints in
\eqref{eq:pre:max}.

The set $\setS$ is said to be {\em feasible set} of the optimization problem if it is the set of all feasible points.
In particular, the {\em optimal solution} of the optimization problem is the one at which the objective function
achieves its maximal value.

One should be noted that for a specific optimization problem, there may exist many different optimal
solutions. Moreover, the optimization problem~\eqref{eq:pre:max} can be written in the following form:
\begin{equation}
  \max_{\kett{x}\in\setS}f(\kett{x}).
\end{equation}

Denote by $\mathcal{L}$ the Lagrange function
\begin{equation}
  \label{eq:pre:L}
  \mathcal{L}(\kett{x},\kett{\mu},\nu)
  = f(\kett{x})  - \sum_{i=1}^m\mu_i \langle \kett{c}_i,\kett{x}\rangle -\nu (\braa{x}W\kett{x}-r),
\end{equation}
where $\kett{\mu} = (\mu_1,\mu_2,\ldots,\mu_m)^\intercal\in\bbR^m$ and $\nu\in\bbR$.  Then we have the well-known KKT
condition for the optimization problem.

\begin{lemma}[First order KKT condition~\cite{Karush_2013}]
  \label{lem:KKT}
  Suppose the optimization problem is defined by \eqref{eq:pre:max} with Lagrange function defined by
  \eqref{eq:pre:L}. If $\kett x^*$ is a optimized solution to the problem, then there exist $\mu_i^*\in\bbR$ and
  $\nu^*\in\bbR$ such that
  \begin{equation}
    \left. \nabla \mathcal{L}(\kett{x},\kett\mu,\nu) \right|_{(\kett{x}^*,\kett\mu^*,\nu^*)}=0,
  \end{equation}
  that is,
  \begin{subequations}
    \begin{align}
      \nabla f(\kett{x}^*) -\sum_{i=1}^m\mu_i^* \kett{c}_i -2\nu^* W\kett{x}^*& = 0,\label{KKT:a}\\[-8pt]
      \langle \kett{c}_i, \kett{x}^*\rangle & =0,\\[2pt]
      \langle \kett{x}^*,W\kett{x}^*\rangle  & =r.
    \end{align}
  \end{subequations}
\end{lemma}


The following two lemmas about the necessary conditions of the optimization problem are useful.

\begin{lemma}\label{lem:y:square}%
Suppose $f(\kett{x})$ is a degree-two homogeneous function where $\kett{x}\in\bbR^N$.  The optimization problem is
  defined as follows:
  \begin{equation}\label{eq:opt1}
    \begin{array}{rl}
      \max & \quad f(\kett{x})+ \sum_{i=1}^M \xi_iy_i^2\\
      \mathrm{s.t. }& \quad
                      \begin{cases}
                        \langle \kett c_i,\kett{x}\rangle & =0, \\
                        \sum_{i=1}^N\tau_i~x_i^2+\sum_{i=1}^M \omega_iy_i^2 &= r,
                      \end{cases}
    \end{array}
  \end{equation}
  where $\kett{y}\in\bbR^M$, $\kett{c_i}\in\bbR^N,\;i=1, \ldots, m,$, $\xi_i\in\bbR $, $\tau_i~>0$, and $\omega_i>0 $.
  Denote by $\setS$ the feasible set of the optimization problem.  Let $(\kett{x}^*,\kett y^*)$ be an optimal
  solution. Then
  \begin{equation}\label{eq:2.2}
    \max_{(\kett{x},\kett{y})\in\setS}f(\kett{x})+ \sum_{i=1}^M \xi_iy_i^2= \eta r, \quad \text{or} \quad \kett y^*=0.
  \end{equation}
  where $\eta = \max\{\frac{\xi_i}{\omega_i},i=1,2,\ldots,M\}$.
\end{lemma}

\begin{lemma}\label{lem:y:same}
  Suppose the optimization problem is defined as follows:
  \begin{equation}\label{eq:opt2}
    \begin{array}{rl}
      \max & \quad f(\kett{x})\\
      \text{s.t.}& \quad
                   \begin{cases}
                     \langle \kett{c},\kett{x}\rangle +\sum_{i=1}^My_i &=0,\\
                     \langle \kett{d}_i,\kett{x}\rangle & =0,   \\
                     \sum_{i=1}^N\tau_i~x_i^2+\sum_{i=1}^My_i^2 & =r,
                   \end{cases}
    \end{array}
  \end{equation}
  where $\kett{x},\;\kett{c}\in\bbR^N,\;\kett{d}_i\in\bbR^N,\; i=1, \ldots, m,,\;\kett{y}\in\bbR^M$ and $\tau_i>0 $.  In
  addition, $f$ is degree-2 homogeneous and $f(\kett{x}^{0})>0$ for some feasible point $\kett{x}^0$.  Then, the maximal
  value of the objective function $f$ is achieved only when
  \begin{equation}
    y_1=y_2=\cdots=y_M.
  \end{equation}
\end{lemma}
For the sake of conciseness, we move the proofs of the above lemmas to~\cref{appendix:y:square,appendix:y:same},
respectively.

\section{Proof of the Main Result}
\label{sec3}
In this section, we will prove our main result \cref{main-result1}, that is, the validity of \cref{cj:high_var}
when one of $A$ and $B$  is normal and the other is arbitrary.
 With a locally unitary similarity, we can assume that $A$ is diagonal.
 Therefore, in the following development of this
  section we always assume that
\begin{equation}\label{A}
A = \mathrm{diag}(a_1,a_2,\ldots,a_d) \in
  \bbC^{d\times d}.
\end{equation}
Thus, $X$ can be written as a direct sum
\begin{equation} X = \oplus_{i=1}^{d}(a_i\I + B). \end{equation}

In order to prove Theorem \ref{main-result1}, we prove the following theorem first, which can be regarded as the real
version of Theorem \ref{main-result1}.

\

\begin{theorem}\label{main-result2}
  Conjecture \ref{cj:high_var} holds when
  \begin{equation} A = \mathrm{diag}(a_1,a_2,\ldots,a_d) \in \bbR^{d\times d}, \quad\quad B\in \bbR^{d\times d}.\end{equation}
\end{theorem}

\

Note that the set of singular values of $X$ consists of all the singular values of each block $a_i\I + B$. Then the
problem of proving Theorem \ref{main-result2} naturally split into two different cases:
\begin{itemize}
  \item [Case~1]:  the largest two singular values of $X$ come from the same block, say $a_1\I+B$;
  \item [Case~2]:  the largest two singular values of $X$ come from two different blocks, say $a_1\I + B$ and $a_2\I + B$.
\end{itemize}

In the following we discuss these two cases separately.

\subsection{Case~1: the Largest Two Singular Values of $X$ Come From the Same Block}

We have the following result for this case.

\begin{lemma}%
  \label{lem:anyB}
  Let $X = A\otimes \I + \I \otimes B \in\cX_d\,(d\geqslant 4)$ with  $A=\mathrm{diag}(a_1,a_2,\ldots,a_d)$. If
  \begin{equation}
    \sigma_1(X) = \sigma_1(a_1\I + B),\quad \quad \sigma_2(X) = \sigma_2(a_1\I + B),
  \end{equation}
  then
  \begin{equation} \tar \leqslant \frac{3d-4}{d^2}.  \end{equation}
\end{lemma}

\begin{proof}
  Denote by $B= (b_{ij})_{i,j=1}^d$. Note that
  \begin{subequations}
    \begin{align}
      \sigma_i^2(X)& = \sigma_i^2(a_1\I + B)
                    = \lambda_i(a_1^2 \I+ a_1(B+\trans B)+B\trans B)  \\
                   & = a_1^2 +
                     \lambda_i(a_1(B+\trans B)+B\trans B),\quad i=1,2.
    \end{align}
  \end{subequations}
  Denote by  $\kett{\phi}$ and $\kett{\psi}$  the two unit eigenvectors (up to a phase multiplication) associated with the first and second largest
  eigenvalues of $a_1(B+\trans B)+B\trans B$, respectively.  With an orthogonal similarity, we can assume that
  $\kett{\phi}=\kett{e}_1$ and $\kett{\psi}=\kett{e}_2$, i.e., replace $B$ with $ UB\trans U$ by some orthogonal
  operator $U$. Hence,
  \begin{align}
    \tar
    & = 2a_1^2 +\trace{(a_1(B+\trans B)+B\trans B)(\kett{\phi}\braa{\phi} + \kett{\psi}\braa{\psi})}\\
    & = 2a_1^2 + 2a_1(b_{11}+b_{22})+ \sum_{i=1}^d(b_{1i}^2+b_{2i}^2)\\
    & = (a_1+b_{11})^2+(a_1+b_{22})^2+\sum^2_{i=1}\sum_{j\neq i}b_{ij}^2.\label{eq:24}
  \end{align}
  Consider $A$ and $\mathrm{diag}(B)$, they are both normal and traceless, by
  Corollary~\ref{lem:normal2}, we have
  \begin{equation}
    (a_1+b_{11})^2+(a_1+b_{22})^2\leqslant\frac{3d-4}{d}\sum_{i=1}^d(a_i^2+b_{ii}^2).
  \end{equation}
  Therefore, by \eqref{eq:24}, we have
  \begin{subequations}
    \begin{align}
       \tar &\leqslant \frac{3d-4}{d}\sum_{i=1}^d(a_i^2+b_{ii}^2)+ \sum^2_{i=1}\sum_{j\neq i}b_{ij}^2\\
      &  = \frac{3d-4}{d}( \frac{1}{d} - \sum_{i\neq j}b_{ij}^2)+ \sum^2_{i=1}\sum_{j\neq i}b_{ij}^2\\
      & \leqslant \frac{3d-4}{d^2} - \frac{2d-4}{d}\sum_{i\neq j}b_{ij}^2\\
      & \leqslant  \frac{3d-4}{d^2},
    \end{align}
  \end{subequations}
  which completes our proof.
\end{proof}

\subsection{Case~2: the Largest Two Singular Values of $X$ Come From Two Different Blocks}

We shall consider the case when $B\in \cP$ first, where $B$ is the direct sum of blocks of at most dimension two. After
that, we can extend this result to the general $B$.

Recall that if $B\in\cP$, then $B$ is the direct sum of several square matrices, i.e.,
\begin{equation}
  B = \oplus_{i=1}^l B_i, \quad \mathrm{dim}(B_i) \leqslant 2,\quad \sum_{i=1}^l \mathrm{dim}(B_i)=d.
\end{equation}

\begin{lemma}%
  \label{lem:22B}
  Suppose $X = A\otimes\I+\I \otimes B$ with $A, B\in \bbR^{d\times d}\;(d\geqslant 4)$, $\trace{A}= \trace{B} = 0$,
  $ \norm{A}_F^2+\norm{B}_F^2 = \frac{1}{d}$, $A=\mathrm{diag}(a_1,a_2,\ldots,a_d)$ and
  $B\in \cP$.  If the largest two singular values of $X$ come from two different blocks, i.e.,
  \begin{equation}
    \sigma_1(X) = \sigma_1(a_1\I + B),\quad \quad \sigma_2(X) = \sigma_1(a_2\I + B),
  \end{equation}
  then
  \begin{equation} \tar \leqslant \frac{3d-4}{d^2}.  \end{equation}

\end{lemma}

\begin{proof} Lemma \ref{lem:22B} is proved in Appendix C. \end{proof}

Now we extend Lemma \ref{lem:22B} to general $B$.

\begin{lemma} \label{lem:from_two} Suppose $X = A\otimes\I+\I \otimes B$  with 
  $A, B\in \bbR^{d\times d}\;(d\geqslant 4)$, $\trace{A} = \trace{B} =0$, $ \norm{A}_F^2+\norm{B}_F^2 = \frac{1}{d}$, and
  $A=\mathrm{diag}(a_1,a_2,\ldots,a_d)$.  If the largest two singular values of $X$ come from two different blocks,
  i.e.,
  \begin{equation}
    \sigma_1(X) = \sigma_1(a_1\I + B) \text{ and } \sigma_2(X) = \sigma_1(a_2\I + B),
  \end{equation}
  then
  \begin{equation} \tar \leqslant \frac{3d-4}{d^2}.  \end{equation}
\end{lemma}

\begin{proof}
  Let
  \begin{subequations}
  \begin{align}
    Y_i &= (a_i\I + B)\trans{(a_i\I + B)}
        = a_i^2\I + a_i(B+\trans B)+B\trans B,\quad i=1,2.
  \end{align}
  \end{subequations}
  Then $\sigma_1^2(X) = \lambda_1(Y_1)$ and $\sigma_2^2(X) = \lambda_1(Y_2)$.  Suppose $\kett{\phi}$ and $\kett{\psi}$
  are the unit eigenvectors associated with the largest eigenvalues of $Y_1$ and $Y_2$, respectively.  With an
  orthogonal similarity of $B$, i.e., replace $B$ with $UB\trans U$ by some orthogonal operator $U$, we can assume that
  \begin{equation}
    \kett{\phi} = \kett{e}_1,\quad\kett{\psi} = \cos\theta\kett{e}_1+\sin\theta\kett{e}_2,\quad\theta\in[0,2\pi].
  \end{equation}
  Furthermore, applying an orthogonal similarity on $B$ with the first two dimensions kept the same, we can further
  assume
  \begin{equation} b_{1j}=0,\;j\geqslant 4, \quad \text{and}\quad
  b_{ 2j}=0,\;j\geqslant 5.
  \end{equation}
  Let $h = \tar$. Then,
  \begin{align}
    h& = \lambda_1(Y_1)+ \lambda_1(Y_2)
     =\braa{\phi}Y_1\kett{\phi}+\braa{\psi}Y_2\kett{\psi}
     =h_1+ h_2+b_{24}^2\sin^2\theta,
  \end{align}
  where
  \begin{multline}  h_1 = \sin(2\theta)\left(a_2(b_{12}+b_{21})+b_{11}{b}_{21}+b_{12}{b}_{22}\right) 
    +(1+\cos^2\theta)b_{12}^2 \\+ \cos^2\theta(a_2+b_{11})^2 +(a_1+b_{11})^2+\sin^2\theta((a_2+b_{22})^2+ b_{21}^2),
  \end{multline}
  and
  \begin{equation}
    h_2=b_{13}^2+(\cos\theta b_{13}+\sin\theta b_{23})^2.
  \end{equation}
  Let
  \begin{equation}
    \mathcal{B} = \{ b_{ij}: i\neq j \&\; (i,j)\not\in \{(1,2),(2,1),(1,3),(2,3),(2,4)\}\},
  \end{equation}
  be the set where the parameters are not involved in the objective function $h$ and the related linear constraints.
  Hence, we can replace $b_{ij}\in\mathcal{B}$ by $0$'s and multiply a scalar (denoted as $\beta$) to all the rest
  variables simultaneously in order to satisfy the related quadratic constraint.  It is easy to see that $\beta >1$ if
  some of $b_{ij}\in\mathcal{B}$ are nonzero.  In this way, the linear constraints are also satisfied and the value of
  function $h$ is replaced by
  \begin{equation}
    \beta^2h>h,
  \end{equation}
  since $\beta>1$ and $h$ is a positive degree-$2$ homogeneous function. We now find another feasible point such that
  $h$ achieves a greater value.  Hence, it suffices to assume
  \begin{equation}
    b_{ij}=0,\quad \forall i\neq j \in \mathcal{B}.
  \end{equation}
  Note that to find the maximal value of $h$
  is equivalent to solve the following optimization problem:
  \begin{equation}\label{12:opt1}
    \begin{array}{rl}
      \max & \quad h\\
      \text{ s.t. }&
                     \begin{cases}
                       \sum\limits_{i=1}^d a_i= 0,\\
                       \sum\limits_{i=1}^d b_i= 0,\\
                       \sum\limits_{i=1}^d\!a_i^2+\sum\limits_{i,j\not\in \mathcal{B}}b_{ij}^2 = \frac{1}{d}.
                     \end{cases}
    \end{array}
  \end{equation}
  Let
  \begin{equation}
    b_{13} = t\cos\alpha,\quad b_{23} = t\sin\alpha,\quad \alpha\in[0,2\pi].
  \end{equation}
  Then
  \begin{equation}
    h_2 =t^2(\cos^2\alpha +\cos^2(\alpha-\theta)).
  \end{equation}

    Therefore, the original optimization problem~\eqref{12:opt1} is reduced to the following one:
    \begin{equation}
      \begin{array}{rl}
        \max & \quad f(\kett{x}) +  \xi_1 y_1^2+\xi_2 y_2^2\\
        \text{s.t.}&\quad
                     \begin{cases}
                       \langle \kett{c_i},\kett{x}\rangle & = 0,\\
                       \sum_{i=1}^{2d+2} \tau_i~|x_i|^2 + \omega_1 y_1^2+ \omega_2 y_2^2 &= r,
                     \end{cases}
      \end{array}
    \end{equation}
    where
    \begin{equation}
      \begin{aligned}
        \kett{x}&  = (a_1,a_2,\ldots,a_d,b_{11}, b_{22}, \ldots,b_{dd},b_{12},b_{21})^\intercal,\\
        \kett{c}_1& = (\underbrace{1,1,\ldots,1}_{d \text{ times}},0,0,\ldots,0, 0, 0)^\intercal\in\bbR^{2d+2},\\
        \kett{c}_2& = (0,0,\ldots,0,\underbrace{1,1,\ldots,1}_{d \text{ times}},0,0)^\intercal\in\bbR^{2d+2},\\
      \end{aligned}
    \end{equation}
    and
    \begin{align}
      \kett{y} =(b_{24},t)^\intercal,
      \quad \tau_i =\omega_i = 1,
      \quad 
      \xi_1 = \sin^2\theta,
      \quad \xi_2=\cos^2\alpha +\cos^2(\alpha-\theta),
      \quad r = \frac{1}{d}.
    \end{align}
    In particular,
    \begin{multline}
       f(\kett{x}) = (x_1+x_{d+1})^2
        + \cos^2\theta(x_2+x_{d+1})^2 
        +\sin^2\theta((x_2+x_{d+2})^2+ x_{2d+2}^2)
        \\ +(1+\cos^2\theta) x_{2d+1}^2
         +\sin(2\theta)\left( x_{d+1}{x}_{2d+2}+x_{2d+1}{x}_{d+2} 
           +x_2(x_{2d+1}+x_{2d+2})\right).
    \end{multline}
  Let
  \begin{equation}
    \eta = \max\{ \frac{\xi_i}{\omega_i},i=1,2\}.
  \end{equation}
  It is eay to see that
  \begin{equation}\eta =\max\{\sin^2\theta,\cos^2\alpha + \cos^2(\alpha-\theta)\} \leqslant 2.\end{equation}
  Note that $f(\kett{x})$ is degree-$2$ homogeneous.  Therefore, according to Lemma~\ref{lem:y:square}, we have either
  \begin{align}
    \max \tar &=\max  ( h_1 +h_2+ b_{24}^2\sin^2\theta)
             = \frac{\eta}{d} \leqslant \frac{2}{d}\leqslant \frac{3d-4}{d^2},
  \end{align}
  or $\tar$ achieves the maximal value only when $\kett{y}=0$, i.e., $b_{13}=b_{23}=b_{24}=0$.

  For the latter case, the matrix $B$ has the shape as follows:
  \begin{equation}
    B = \begin{bmatrix}
      b_{11}&b_{12}   \\
      b_{21}& b_{22} \\
    \end{bmatrix}
    \oplus \mathrm{diag}(b_{33},b_{44},\ldots,b_{dd}),
  \end{equation}
  and then $B \in \cP$. Hence, Lemma \ref{lem:from_two} follows directly from Lemma~\ref{lem:22B}.
\end{proof}

Theorem \ref{main-result2} can be proved based on Lemmas~\ref{lem:anyB} and \ref{lem:from_two} as follows.

\

\begin{proof}[Proof of Theorem \ref{main-result2}] $X$ is a direct sum of $d$ blocks. Then, the set of singular values
  of $X$ consists of all the singular values of all the different blocks.  By Lemma~\ref{lem:anyB}, the inequality
  (\ref{complex}) holds when the largest two singular values of $X$ come from a single block. On the other hand, by
  Lemma~\ref{lem:from_two}, the inequality (\ref{complex}) also holds when the largest two singular values of $X$ come
  from two different blocks. Hence Theorem \ref{main-result2} follows.
\end{proof}

We are now ready to prove Theorem \ref{main-result1}.

\

\begin{proof}[Proof of Theorem \ref{main-result1}] Since $A=\mathrm{diag}(a_1,a_2,\ldots,a_d)$ and thus
  \begin{equation} X = \oplus_{i=1}^d (a_i\I + B). \end{equation} Let $Y= X\adj{X}$. Then
  \begin{equation} Y = \oplus_{i=1}^d\left[ (a_i\I+ B)\adj{(a_i\I + B)}\right].
  \end{equation}
  Suppose $\kett{\phi}_i$ are the two unit eigenvectors of $Y$ corresponding to its largest two eigenvalues.  Note that
  $Y$ is the direct sum of $d$ block matrices, its eigenvalues consist of all the eigenvalues of every block.  Moreover,
  the eigenvectors of $Y$ have the shapes:
  \begin{equation}
    \kett{\phi}_i = \kett{e}_i\otimes \kett{y}_i,\quad \kett{e}_i,\kett{y}_i\in\bbC^d,\;i=1,\ldots,d^2,\;
  \end{equation}
  where $\kett{e}_i,i=1,\ldots,d$ is the natural basis in $\bbC^d$.  In particular, we will have the two different cases
  up to an index permutation:
  \begin{enumerate}[label = (\roman*)]
    \item $\kett{\phi}_1 = \kett{e}_1\otimes \kett{y}_1$ and $\kett{\phi}_2 = \kett{e}_1\otimes \kett{y}_2$;
    \item $\kett{\phi}_1 = \kett{e}_1\otimes \kett{y}_1$ and $\kett{\phi}_2 = \kett{e}_2\otimes \kett{y}_2$.
  \end{enumerate}

  In case (i): we have $\langle \kett{y}_1,\kett{y}_2\rangle =0$.  Hence, there exists a unitary operator $U$ such that
  \begin{equation}
    U(\kett{y}_1,\kett{y}_2) = (\kett{e}_1,\kett{e}_2).
  \end{equation}

  In case (ii): we can find a unitary operator $U_1$ such that
  \begin{equation}
    U_1(\kett{y}_1,\kett{y}_2) = (\kett{e}_1,c_1\kett{e}_1+c_2\kett{e}_2),
  \end{equation}
  where $c_1,c_2\in\bbC$. Suppose
  \begin{align}
    c_1 &= r_1e^{\iu\theta_1},
    \quad \quad c_2 = r_2e^{\iu\theta_2},\quad
    \theta_1,\theta_2\in[0,2\pi], \quad
    r_1^2+r_2^2=1,
  \end{align}
  and let
  \begin{equation}
    U_2 = \mathrm{diag}(1,e^{\iu(\theta_1-\theta_2)},1,\ldots,1)\in\bbC^{d\times d}.
  \end{equation}
  Then
  \begin{equation}
    U(\kett{y}_1,\kett{y}_2) = (\kett{e}_1,e^{\iu\theta_1}(r_1\kett{e}_1+r_2\kett{e}_2)),
  \end{equation}
  where $U=U_2U_1$.  Note that if $\kett{x}$ is an eigenvector of a Hermitian matrix, then $e^{\iu\theta}\kett{x}$ is
  also the eigenvector corresponding to the same eigenvalue. So, if we replace $B$ by $UB\adj{U}$, then $\kett{e}_1$ and
  $r_1\kett{e}_1+r_2\kett{e}_2$ are the two eigenvectors of $Y$ associated with its largest two eigenvalues and these
  two eigenvectors are real.

  Therefore, we can denote by $\kett{\phi}_i \in \bbR^{d^2} (i=1,2)$ the unit eigenvectors of $Y$ associated with its
  largest two eigenvalues, which are real. Then
  \begin{equation}
    \langle \kett{\phi}_1,\kett{\phi}_2\rangle =0.
  \end{equation}

  Let $X = X_1+ \iu X_2$ with $X_1, X_2\in \bbR^{d^2\times d^2}$. Then,
  \begin{equation}
    Y = X_1X_1^\intercal + X_2X_2^\intercal + \iu( X_2 X_1^\intercal - X_1X_2^\intercal).
  \end{equation}
  There
  \begin{align}
    \sigma_i^2(X)&= \kett{\phi}^\intercal_i Y\kett{\phi}_i
                  = \kett{\phi}^\intercal_i \left(X_1X_1^\intercal + X_2X_2^\intercal \right)\kett{\phi}_i
                       + \iu\left(\kett{\phi}^\intercal_i (X_2 X_1^\intercal - X_1X_2^\intercal)\kett{\phi}_i\right).
  \end{align}
  Since $\sigma_i^2(X)\geqslant 0$, $\kett{\phi} _i\in\bbR^{d^2}$, and $X_1,X_2\in\bbR^{d^2\times d^2}$, we have
  \begin{align}
    \sigma_i^2(X)& = \kett{\phi}^\intercal_i \left(X_1X_1^\intercal + X_2X_2^\intercal \right)\kett{\phi}_i 
                 =
                   \kett{\phi}^\intercal_i X_1X_1^\intercal \kett{\phi}_i+ \kett{\phi}^\intercal_i X_2X_2^\intercal \kett{\phi}_i.
  \end{align}
  Moreover,
  \begin{align}
    \max_{\stackrel{\langle \kett\psi_1,\kett{\psi}_2\rangle=0}{\norm{\psi_1}=\norm{\psi_2}=1 }}\quad &\sum_{i=1}^2\kett{\psi}^\intercal_i X_iX_i^\intercal
      \kett{\psi}_i 
    = \lambda_1(X_iX_i^\intercal)+ \lambda_2(X_iX_i^\intercal) 
    = \sigma_1^2(X_i)+\sigma_2^2(X_i).
  \end{align}
  Thus, we have
  \begin{align}\label{X} \tar \leqslant
    \sigma_1^2(X_1)+\sigma_2^2(X_1) + \sigma_1^2(X_2)+\sigma_2^2(X_2).
  \end{align}
  Note that
  \begin{equation} X_1 = \Real{X} = \Real{A}\otimes \I + \I \otimes \Real{B},
  \end{equation}
  which also satisfies the conditions:
  \begin{align}
    \trace{\Real{A}}&= \trace{\Real{B}}=0, \\
     \Real{A}&=\Real { \mathrm{diag}(a_1,a_2,\ldots,a_d)} \in \bbR^{d\times d}, \\
       \Real{B}& \in \bbR^{d\times d}.
  \end{align}
  Let
  \begin{equation} s_1=[d (\norm{\Real{A}}_F^2+\norm{\Real{B}}_F^2)]^{-1/2}. \end{equation} Then we have by Theorem \ref{main-result2}
  \begin{equation} \sigma_1^2(s_1 X_1)+ \sigma_2^2(s_1 X_1) \leqslant \frac{3d-4}{d^2}, \end{equation} which gives
  \begin{equation}
    \sigma_1^2(X_1)+ \sigma_2^2(X_1) \leqslant \frac{3d-4}{d}(\norm{\Real{A}}_F^2+\norm{\Real{B}}_F^2).
  \end{equation}

  Similarly,
  \begin{equation} X_2 = \Ima{X} = \Ima{A}\otimes \I + \I \otimes \Ima{B},
  \end{equation}
  which also satisfies the conditions:
  \begin{align}
    \trace{\Ima{A}}&= \trace{\Ima{B}}=0,\\
         \Ima{A}&=\Ima { \mathrm{diag}(a_1,a_2,\ldots,a_d)} \in \bbR^{d\times d},\\
    \Ima{B} &\in \bbR^{d\times d}.
  \end{align}
  Let
  \begin{equation} s_2=[d (\norm{\Ima{A}}_F^2+\norm{\Ima{B}}_F^2)]^{-1/2}.
  \end{equation} Then we have by Theorem \ref{main-result2}
  \begin{equation} \sigma_1^2(s_2 X_2)+ \sigma_2^2(s_2 X_2) \leqslant \frac{3d-4}{d^2},
  \end{equation} which gives
  \begin{equation}
    \sigma_1^2(X_2)+ \sigma_2^2(X_2) \leqslant \frac{3d-4}{d}(\norm{\Ima{A}}_F^2+\norm{\Ima{B}}_F^2).
  \end{equation}
  Therefore, by \eqref{X}, we have
  \begin{align}
    \tar  &\leqslant \frac{3d-4}{d}(\norm{\Real{A}}_F^2+\norm{\Real{B}}_F^2)
      + \frac{3d-4}{d}(\norm{\Ima{A}}_F^2+\norm{\Ima{B}}_F^2)   \\
    &  = \frac{3d-4}{d}(\norm{\Real{A}}_F^2+\norm{\Real{B}}_F^2 
     +\norm{\Ima{A}}_F^2+\norm{\Ima{B}}_F^2 )  \\
    & = \frac{3d-4}{d}(\norm{A}_F^2+\norm{B}_F^2)  \\
    & = \frac{3d-4}{d^2},
  \end{align}
  i.e., \eqref{complex} holds.
\end{proof}

%
%
%
\section{Conclusion}\label{sec:concluding}

We have studied Conjecture \ref{cj:high_var} related to the entanglement distillability problem, one of the fundamental
problems in quantum information theory.  In Refs.~\cite{pphh2010,ysdis2018}, this conjecture is answered partially when
$A,B$ are both normal. In our main result--Theorem \ref{main-result1}, we solved this open problem when one of $A$ and
$B$ is normal and the other one is arbitrary utilizing the techniques of matrix analysis and optimization theory.
Our Results-Theorems \ref{main-result1} and \ref{main-result11} and Corollaries \ref{main-result2} and
\ref{main-result22} make significant progress on Conjecture \ref{cj:main-4} and \ref{cj:high_var}, namely the
two-undistillability of the one-copy undistillability NPT state.

%
%
%

\appendices
\section{Proof of Lemma~\ref{lem:y:square}}\label{appendix:y:square}

Before proving Lemma~\ref{lem:y:square}, the following result is needed.
\begin{lemma}\label{lem:A1}
  Suppose
  \begin{align}
    \kett{a}&= (a_1,a_2,\ldots,a_N)^\intercal, \\
    \kett{b}&= (\,b_1,b_2,\ldots,\,b_N)^\intercal, \\
    \kett{x} &= (x_1,x_2,\ldots,x_N)^\intercal,
    \end{align}
  where $a_i,x_i\in\bbR$  and $b_i>0$.  Then for any nonzero $\kett x \in \bbR^N$,
  \begin{align}
    \frac{\sum_{i=1}^N a_ix_i^2}{\sum_{i=1}^Nb_ix_i^2}\leqslant \max\left\lbrace\frac{a_i}{b_i},\; i=1,2,\ldots,N\right\rbrace.
    \end{align}
  \end{lemma}

\begin{proof}
  Denote by $\kett{c} = (\sqrt{b_1},\sqrt{b_2},\ldots,\sqrt{b_N})^\intercal$, $A=\mathrm{diag}(\kett{a})$,
  $B= \mathrm{diag}(\kett{b})$, and $C = \mathrm{diag}(\kett{c})$.  Then $B= C^2$ and
  \begin{align}
    \frac{\sum_{i=1}^N a_ix_i^2}{\sum_{i=1}^Nb_ix_i^2}
    & =\frac{ \langle \kett{x},A\kett{x}\rangle }{\langle \kett{x},B \kett{x}\rangle}
                 =\frac{ \langle \kett{x},A\kett{x}\rangle }{\langle \kett{x},C^2 \kett{x}\rangle}
    = \frac{ \langle C\kett{x},(C^{-1}AC^{-1})C\kett{x} \rangle }{\langle C \kett{x},C \kett{x}\rangle }\\
    &= \frac{ \langle \kett{y},(C^{-1}AC^{-1})\kett{y} \rangle }{\langle \kett{y}, \kett{y}\rangle }  \\
    & \leqslant \lambda_1(C^{-1}AC^{-1}) \\
    & = \lambda_1(AB^{-1}) \\
    &=\max\left\lbrace\frac{a_i}{b_i},\; i=1,2,\ldots,N\right\rbrace,
  \end{align}
  where $\kett{y}=C\kett{x}$.
\end{proof}

\begin{proof}[Proof of  Lemma~\ref{lem:y:square}]
  Suppose $f(\kett{x})+ \sum_{i=1}^M \xi_iy_i^2$ achieves its maximal value at $(\kett{x}^*,\kett y^*)$.  Let
  \begin{equation}\label{eq:21}
    \zeta = f(\kett{x}^*) + \sum_{i=1}^M\xi_i|y_i^*|^2 = \max_{(\kett{x},\kett{y})\in\setS}f(\kett{x})+ \sum_{i=1}^M \xi_iy_i^2.
  \end{equation}
  Without the loss of generality, we can assume $\eta = \frac{\xi_1}{w_{1}}$.  Then we must have
  \begin{equation}
    \zeta \geqslant \eta r,
  \end{equation}
  since the value $\eta r$ can be achieved at
  $\kett{x} = 0,\;y_1= \frac{\sqrt{r}}{\sqrt{w_{1}}}$ and $y_i=0,\;i\geqslant 2$.

  We prove this lemma by contradiction.

  Assume that result of Lemma~\ref{lem:y:square}, i.e., \eqref{eq:2.2} does not hold:
  \begin{equation}\label{eq:A:11}
    \kett y^*\neq 0 \text{ and } \zeta> \eta r.
  \end{equation}
  Let
  \begin{align}
    r_1 &= \sum_{i=1}^N\tau_i|x_i^*|^2,\quad\text{and}\quad 
    r_2 = \sum_{i=1}^M \omega_i|y_i^*|^2.
  \end{align}

  By the quadratic constraint of the optimization problem~\eqref{eq:opt1}, we must have
  \begin{equation}
    r_1+r_2 = r,\quad r_1\geqslant 0,\;r_2>0.
  \end{equation}
  If $r_1=0$, then $\kett{x}=0$. Let $u_i = \sqrt{w_i}|y_i^*|$, then
  \begin{equation}
    \begin{aligned}
      \zeta &= f(\kett{x}^*) + \sum_{i=1}^M\xi_i|y_i^*|^2 
      = \sum_{i=1}^M\xi_i|y_i^*|^2 =\frac{\sum_{i=1}^M
        \frac{\xi_i}{\omega_i}u_i^2}{\sum_{i=1}^Mu_i^2}r.
    \end{aligned}
  \end{equation}
  According to Lemma~\ref{lem:A1}, we have
  \begin{equation}
    \zeta\leqslant r \max\{\frac{\xi_i}{\omega_i},i=1,\ldots,M\}=\eta r,
  \end{equation}
  which contradicts \eqref{eq:A:11}.

  Hence, we assume $r_1>0$.  Define the function
  \begin{equation}
    h(t) = f(t\kett{x}^*) + \sum_{i=1}^M\xi_i|\alpha_t y_i^*|^2,
  \end{equation}
  where
  \begin{equation}
    \label{eq:appd:con2}
    \alpha_t  = \sqrt{ \frac{r-r_1t^2}{r_2}},\quad t^2\leqslant\frac{r}{r_1}.
  \end{equation}
  Note that \eqref{eq:appd:con2} guarantees the condition
  \begin{equation}
    \sum_{i=1}^N\tau_i  |tx_i^*|^2+ \sum_{i=1}^M\omega_i|\alpha_t y_i^*|^2
    = t^2 r_1+\alpha_t^2r_2=r.
  \end{equation}
  Obviously, $(t\kett{x}^*,\alpha_t\kett y^*)$ also satisfy the linear constraints:
  \begin{equation}
    \langle \kett{c}_i,t\kett{x}^*\rangle =   t\langle \kett{c}_i,\kett{x}^*\rangle =0,\quad i=1,\ldots,m,
  \end{equation}
  So, $(t\kett{x}^*,\alpha_t\kett y^*)$ is also a feasible point of the optimization problem~\eqref{eq:opt1}.
  Specifically,
  \begin{equation}
    h(t) =\beta_1 t^2+\beta_2,
  \end{equation}
  where
  \begin{align}
    \beta_1 &= f(\kett{x}^*)   -\frac{r_1}{r_2}\sum_{i=1}^M\xi_i|y_i^*|^2,\quad\text{and}\quad
    \beta_2 = \frac{r}{r_2}\sum_{i=1}^M\xi_i|y_i^*|^2.
  \end{align}
  If the coefficient $\beta_1$ is positive, then $h(t)$ is strictly increasing with respect to $t$ when $t\geqslant 0$.
  Let
  \begin{equation}
    t_0 = \sqrt{\frac{r}{r_1}}= \sqrt{1+\frac{r_2}{r_1}}>1.
  \end{equation}
  \noindent Further, we have that $h(t_0)>h(1)=\zeta$. In other words, the objective function of the optimization
  problem~\eqref{eq:opt1} reaches a greater value at $(t\kett{x}^*,\alpha_t\kett y^*)$, which contradicts \eqref{eq:21}.

  On the other hand, if the coefficient $\beta_1$ is non-positive, then $h(t)$ is monotonously decreasing (maybe
  constant) with respect to $t$ when $t\geqslant 0$. We have,
  \begin{equation}\label{eq:22}
    \zeta = h(1)\leqslant h(0) = \beta_2=\frac{\sum_{i=1}^M\xi_i|y_i^*|^2}{\sum_{i=1}^M\omega_i|y_i^*|^2}r\leqslant \eta r,
  \end{equation}
  where the last inequality follows from Lemma~\ref{lem:A1}.

  However, \eqref{eq:22} contradicts our assumption~\eqref{eq:A:11}.

  To sum up, we can conclude that
  \begin{equation}
    \max_{(\kett{x},\kett{y})\in\setS}f(\kett{x})+  \sum_{i=1}^N \xi_iy_i^2= \eta r \quad \text{or}\quad  \kett{y}^* = 0.
  \end{equation}
\end{proof}

\section{Proof of  Lemma~\ref{lem:y:same}}
\label{appendix:y:same}
\begin{proof}[Proof of Lemma~\ref{lem:y:same}]
Suppose the feasible point $(\kett{x}^*,\kett{y}^*)$ is an optimal solution of the optimization problem~\eqref{eq:opt2}, that is,
  \begin{equation}\label{eq:11}
   f(\kett{x}^*)=  \max_{(\kett{x},\kett{y})\in\setS}f(\kett{x}) = \zeta.
  \end{equation}

We prove this lemma by contradiction.

  Assume $y_i^*,i=1,\ldots,M$ are not all equal. Specifically, let
  \begin{equation}\label{eq:2}
    y_1^*\neq y_2^*.
  \end{equation}
  Suppose $y_1^*<y_2^*$ ( proof for $y_1^*>y_2^*$ is the same by exchanging the symbols).  Note  that
  \begin{equation}
    |y_1^*+t|^2+|y^*_2-t|^2 = |y_1^*|^2+|y_2^*|^2 - 2(y^*_2-y^*_1)t + 2t^2.
  \end{equation}
  Hence the following inequality holds
  \begin{equation}
    |y^*_1+t|^2+|y^*_2-t|^2< |y^*_1|^2+|y^*_2|^2
  \end{equation}
  when
  \begin{equation}
    0< t < y^*_2-y^*_1.
  \end{equation}
  If we replace $(y_1^*,y^*_2)$ by $(y_1^*+t,y^*_2-t)$, then
  \begin{equation}
    \sum_i\tau_i   |x_i^*|^2 + |y^*_1+t|^2+|y^*_2-t|^2 + \sum_{i=3}^M|y_i^*|^2= \delta_t r,
  \end{equation}
  where
\begin{align}
    \delta_t& = \frac{ \sum\limits_{i=1}^N\tau_i  | x_i^*|^2 + |y_1^*+t|^2+|y^*_2-t|^2+ \sum\limits_{i=3}^M|y_i^*|^2 }
    { \sum\limits_{i=1}^N\tau_i   |x_i^*|^2 + |y_1^*|^2+|y_2^*|^2 + \sum\limits_{i=3}^M|y_i^*|^2}
    = 1- \frac{2(y_2^*-y^*_1)t -2t^2}{r}.
\end{align}
 Hence,
  \begin{equation}
    0< \delta_t<1, \quad \forall \;0< t< y^*_2-y^*_1.
  \end{equation}
  In order to satisfy all the constraints of the optimization
  problem~\eqref{eq:opt2}, we can multiply a positive constant $\beta$,
  which is larger than $1$, to all the parameters. That is, replace
  $\kett{x}^*$ by $\beta \kett{x}^*$ and $(y_1^*,y^*_2,y_3^*,\ldots,y_M^*)$ by
  $(\beta(y_1^*+t),\beta(y_2^*-t),\beta y_3^*,\ldots,\beta y_M^*)$ such that
  \begin{equation}\label{eq:3}
    \sum_{i=1}^N \tau_i  |\beta x_i^*|^2 + |\beta(y_1^*+t)|^2+|\beta(y_2^*-t)|^2+ \sum_{i=3}^M|\beta y_i^*|^2 = r.
  \end{equation}
   Hence, we have
  \begin{align}
    \beta
    &= \frac{\sqrt{r}}{\sqrt{ \sum\limits_{i=1}^N \tau_i  | x_i^*|^2 + |y_1^*+t|^2+|y_2^*-t|^2+ \sum\limits_{i=3}^M|y_i^*|^2}}
    =\frac{1}{\sqrt{\delta_t}}  >1.
  \end{align}
  In addition,
  \begin{equation}\label{eq:10}
    \begin{aligned}
    \langle \kett{c},\beta \kett{x}^*\rangle & + \beta((y_1^*+t)+(y_2^*-t))+\sum_{i=3}^M\beta y_i^*
    = \beta(\langle \kett{c}^*, \kett{x}^*\rangle +\sum_{i=1}^My_i^*) = 0,\\
   \langle \kett{d}_i,\beta\kett{x}^*\rangle & = \beta \langle \kett{d}_i,\kett{x}^*\rangle =0,i=1,2,\ldots,m.
    \end{aligned}
  \end{equation}
  Therefore, combine~\eqref{eq:3} and \eqref{eq:10}, we can know that $(\beta\kett{x}^*,\beta (y_1^*+t),\beta(y_2^*-t),\beta y_3^*,\ldots,\beta y_M^*)$
  is also a feasible point.
  However,
  \begin{equation}
    f(\beta \kett{x}^*) = \beta^2 f(\kett{x}^*)>f(\kett{x}^*)=\zeta,
  \end{equation}
  where the last inequality comes from $f(\kett{x}^*)\geqslant f(\kett{x^0})>0$. This is a contradiction to \eqref{eq:11},
  then the assumption \eqref{eq:2} does not hold.
  Therefore, $f$ achieves the maximal value only when $y_i$ are all  equal.
\end{proof}

\section{Proof of Lemma~\ref{lem:22B}} %
\label{appendix:C} We prove Lemma \ref{lem:22B} in this appendix.

Since
\begin{equation} X = \oplus_{i=1}^{d}(a_i\I + B), \end{equation} and
\begin{equation}
  B = \oplus_{i=1}^l B_i, \quad \mathrm{dim}(B_i) \leqslant 2,\quad \sum_{i=1}^l \mathrm{dim}(B_i)=d,
\end{equation}
we shall consider two different cases, say
\begin{itemize}
  \item[1)] $\sigma_1(X) = \sigma_1(a_1\I + B_1)$ and $\sigma_2(X) = \sigma_1(a_2\I + B_2)$;
  \item[2)] $\sigma_1(X) = \sigma_1(a_1\I + B_1)$ and  $\sigma_2(X) = \sigma_1(a_2\I + B_1)$.
\end{itemize}
We consider Case 1) first.

\begin{lemma}\label{lem:c0}
  Suppose $X = A\otimes\I+\I \otimes B$ with $A,\, B\in \bbR^{d\times d}\;(d\geqslant 4)$, $\trace{A}=\trace{B} =0$, $\norm{A}_F^2+\norm{B}_F^2 = \frac{1}{d}$, $A=\mathrm{diag}(a_1,a_2,\ldots,a_d)$ and
  $B\in \cP$.  If
  \begin{equation} \sigma_1(X) = \sigma_1(a_1\I + B_1), \quad\quad \sigma_2(X) = \sigma_1(a_2\I + B_2), \end{equation} then
  \begin{equation} \tar \leqslant \frac{3d-4}{d^2}.  \end{equation}
\end{lemma}

\begin{proof} We have
  \begin{equation}
    \begin{aligned}
     \! \tar & \leqslant (a_1+\sigma_1(B_1))^2 + (a_2+ \sigma_1(B_2))^2  \\
           & \leqslant 2 ( a_1^2+ a_2^2 + \sigma_1^2(B_1)+ \sigma_1^2(B_2))  \\
           & \leqslant 2 ( \norm{A}_F^2+ \norm{B}_F^2)  \\
           & = \frac{2}{d} \leqslant\frac{3d-4}{d^2}.
    \end{aligned}
  \end{equation}
  Our proof completes.
\end{proof}

We then consider Case 2). First of all, we can prove that $\tar$ achieves the maximal value only when $B_i$
are all diagonal for $i=2,3,\ldots,d$.
\begin{lemma}\label{lem:c1}
 Suppose $X = A\otimes\I+\I \otimes B$ with $A, B\in \bbR^{d\times d}\;(d\geqslant 4)$,
  $\trace{A}=\trace{B}=\nobreak0,\break\; \norm{A}_F^2+\norm{B}_F^2 = \frac{1}{d}$, $A=\mathrm{diag}(a_1,a_2,\ldots,a_d)$ and
  $B\in \cP$.  If
  \begin{equation} \sigma_1(X) = \sigma_1(a_1\I + B_1),\quad \quad \sigma_2(X) = \sigma_1(a_2\I + B_1), \end{equation} Then $\tar$ achieves the
  maximal value only when $B_i,i\geqslant 2$ are all diagonal.
\end{lemma}

\begin{proof}
  If $\mathrm{dim}(B_1)=1$, let
  \begin{equation}
    \kett{x} = (a_1,\ldots a_d,b_{11},\ldots,b_{dd})^\intercal.
  \end{equation}
  If $\mathrm{dim}(B_1)=2$, let
  \begin{equation}
    \kett{x} = (a_1,\ldots a_d,b_{11},\ldots,b_{dd},b_{12},b_{21})^\intercal.
  \end{equation}
  Let $\kett{y}$ be the list of variables consisting of all the off-diagonal entries of $B_i,i\geqslant 2$.  Next,
  define the function $f(\kett{x})$ as
  \begin{equation}\label{eq:14}
    \begin{aligned}
    f(\kett{x}) &= \tar\\
    &=\sigma_1^2(a_1\I + B_1)+\sigma_1^2(a_2\I + B_1).
    \end{aligned}
  \end{equation}
  It is obvious that $f$ is positive degree-two homogeneous.  In order to maximize $f$, it is equivalent to solve the
  following optimization problem:
  \begin{equation}
    \begin{array}{rl}
      \max & f(\kett{x})+\xi_iy_i^2\\
      \mathrm{s.t.} &
                      \begin{cases}
                        \langle \kett{c}_i,\kett{x} \rangle &=0,\quad i=1, 2, \\
                        \sum_i\tau_i x_i^2+\sum_i\omega_iy_i^2&=r,
                      \end{cases}
    \end{array}
  \end{equation}
  where
  \begin{equation}
    \begin{aligned}
      \kett{c}_1 &= (\underbrace{1,1,\ldots,1}_{d \text{~times}},\underbrace{0,0,\ldots,0}_{d\text{~times}},0,\ldots,0)^\intercal,\\
      \kett{c}_2& = (\underbrace{0,0,\ldots,0}_{d
        \text{~times}},\underbrace{1,1,\ldots,1}_{d\text{~times}},0,\ldots,0)^\intercal,
    \end{aligned}
  \end{equation}
  and
  \begin{equation}
    \xi_i = 0,\quad \tau_i  =1,\quad \omega_i = 1,\quad r=\frac{1}{d}.
  \end{equation}
  It follows that
  \begin{equation}\label{eq:12}
    \eta = \max\{\frac{\xi_i}{\omega_i}\}=0.
  \end{equation}
  According to Lemma~\ref{lem:y:square}, we have either
  \begin{equation}
    \max f(\kett{x}) = f(\kett{x})+\xi_iy_i^2 =\eta r =0,
  \end{equation}
  or $f(\kett{x})$ achieves the maximal value only when $\kett{y}=0$, i.e., $B_i,i\geqslant 2$ are all diagonal.
  However, $f(\kett{x})$ is always positive, that is \eqref{eq:12} cannot hold. Hence, our proof completes.
\end{proof}

\begin{lemma}\label{lem:c2}
  Suppose $X = A\otimes\I+\I \otimes B\in\cX_d\;(d\geqslant 4)$ with the conditions that $A, B\in \bbR^{d\times d}$,  $A=\mathrm{diag}(a_1,a_2,\ldots,a_d)$ and
  $B=B_1\oplus\mathrm{diag}(b_{33},b_{44},\ldots,b_{dd}),\;\mathrm{dim}(B_1)=2$.  If
  \begin{equation}
    \sigma_1(X) = \sigma_1(a_1\I + B_1), \quad  \quad  \sigma_2(X) = \sigma_1(a_2\I + B_1),
  \end{equation}
  then $\tar$ achieves the maximal value only when
  \begin{align}\label{eq:13}
    a_{i}&=a_3,\quad i=4,\ldots,d,\\\label{eq:15}
    b_{ii}&=b_{33},\quad i=4,\ldots,d.
  \end{align}
\end{lemma}

\begin{proof}
  We first prove \eqref{eq:13} is necessary for maximizing $\tar$.  Let
  \begin{equation}\label{eq:16}
    \begin{aligned}
      \kett{x} &= (a_1,a_2,b_{11},b_{22},\ldots,b_{dd},b_{12},b_{21})^\intercal\in\bbR^{d+4},\\
      \kett{y} & =(a_3,a_4,\ldots,a_d)^\intercal\in\bbR^{d-2},\\
      \kett{c} & = (1,1,0,0\ldots,0)^\intercal\in\bbR^{d+4},\\
      \kett{d} & = (0,0,1,1,\ldots,1,0,0)^\intercal\in\bbR^{d+4},\\
    \end{aligned}
  \end{equation}
  and
  \begin{equation}
    \tau_i   =1,\; i=1,\ldots,d+4,\qquad      r = \frac{1}{d}.
  \end{equation}
  As in \eqref{eq:14}, we define
  \begin{equation}
    \begin{aligned}
          f(\kett{x}) &= \tar
           =\sigma_1^2(a_1\I + B_1)+\sigma_1^2(a_2\I + B_1).
    \end{aligned}
  \end{equation}
  Similarly, we know that $f(\kett{x})$ is a positive degree-two homogeneous function.

  Then, to maximize $f(\kett{x})$ is equivalent to solve the following optimization problem:
  \begin{equation}
    \begin{array}{rl}
      \max & f(\kett{x})\\
      \mathrm{s.t.} &
                      \begin{cases}
                        \langle \kett{c} ,\kett{x}\rangle + \sum_{i=1}^{d-2}y_i &=0,\\
                        \langle \kett{d}, \kett{x} \rangle &=0,\\
                        \sum_{i=1}^{d+4}\tau_i x_i^2+\sum_{i=1}^{d-2}y_i^2&=r.
                      \end{cases}
    \end{array}
  \end{equation}
  According to Lemma~\ref{lem:y:same}, $f(\kett{x})$ achieves the maximal value only when \eqref{eq:13} holds.

  Similarly, we can prove $f(\kett{x})$ achieves the maximal value only when Eq.\eqref{eq:15} holds by exchanging the
  symbols of $a_i$ and $b_{ii}$ in \eqref{eq:16} for any $i=1,2,\ldots,d$. Therefore, our proof completes.
\end{proof}

In the following part, according to the results in Lemma~\ref{lem:c1} and Lemma~\ref{lem:c2}, we can assume that
\begin{equation}\label{eq:19}
  A= \mathrm{diag}(a_1,a_2,\underbrace{a_3,a_3,\ldots,a_3}_{d-2\text{~times}}),
\end{equation}
and
\begin{equation}\label{eq:20}
  B =B_1\oplus \mathrm{diag}(\underbrace{b_{33},b_{33},\ldots,b_{33}}_{d-2\text{~times}}),
\end{equation}
where
\begin{equation}
  B_1=
  \begin{bmatrix}
    b_{11}&b_{12}\\
    b_{21}&b_{22}
  \end{bmatrix}.
\end{equation}

\begin{lemma}\label{lem:c3}
  Suppose $X = A\otimes\I+\I \otimes B$ with $A, B\in \bbR^{d\times d}\;(d\geqslant 4)$,
  $\trace{A}=\trace{B}=0,\break \norm{A}_F^2+\norm{B}_F^2 = \frac{1}{d}$. Let $A$ and $B$ be defined as in ~\eqref{eq:19}
  and ~\eqref{eq:20}, respectively.  If
  \begin{align}
    \sigma_1(X)& = \sigma_1(a_1\I_2+ B_1), \quad \\
    \sigma_2(X)&  = \sigma_2(a_2\I_2+B_1),
  \end{align}
  then
  \begin{equation}
    \tar \leqslant  \frac{3d-4}{d^2}.
  \end{equation}
\end{lemma}

\begin{proof}
  Denote by $\kett{\phi}\in \bbR^2$ and $\kett{\psi}\in\bbR^2$ the unit eigenvectors corresponding to the largest
  eigenvalues of $(a_1\I_2 + B_1)\trans{(a_1\I_2 + B_1)}$ and $(a_2\I_2 + B_1)\trans{(a_1\I_2 + B_1)}$, respectively.
  With a locally orthogonal similarity, we can assume that
  \begin{equation}
    \kett{\phi} = \kett{e}_1, \quad \kett{\psi } = \cos\theta\kett{e}_1 + \sin\theta \kett{e}_2,\quad \theta\in[0,2\pi].
  \end{equation}
  Hence,
  \begin{multline}
    h(t)=\tar = (a_1+b_{11})^2+   (1+\cos^2\theta )b_{12}^2   
   + \sin^2\theta b_{21}^2+\cos^2\theta(a_2+b_{11})^2\\ +\sin^2\theta(a_2+b_{22})^2
    + \sin(2\theta)( a_2b_{12}+a_2b_{21} + b_{11}b_{21}+ b_{12}b_{22}) .
  \end{multline}
  Then $\tar$ can be regarded as a function of $\theta$.

  Let
  \begin{eqnarray}
    t & =&\tan\theta,  \\
    k & =&
           (a_1+b_{11})^2
           +b_{21}^2+b_{12}^2+
           (a_2+b_{22})^2,    \\
    m & = &b_{12}^2-b_{21}^2+(a_2+b_{11})^2-(a_2+b_{22})^2,  \label{eq:m} \\
    n & =&2b_{12}(a_2+b_{22})+2b_{21}(a_2+b_{11}). \label{n}
  \end{eqnarray}
  A simple calculation yields
  \begin{equation}
    \tar  =  h(t)= k+\frac{m+nt}{1+t^2}.
    \label{eq:absb31-6}
  \end{equation}
  If $n=0$, we have
  \begin{subequations}
    \label{eq:29}
    \begin{align}
      \tar & = h\leqslant k+m
            = |a_1+b_{11}|^2+|a_2+b_{11}|^2+2b_{12}^2\\
           & \leqslant \frac{3d-4}{d}\left( \sum_{i=1}^d(a_i^2+b_{ii}^2)\right) + 2b_{12}^2\label{eq:28}\\
           & = \frac{3d-4}{d}( \frac{1}{d} - b_{12}^2-b_{21}^2) + 2b_{12}^2\\
           & = \frac{3d-4}{d^2} - \frac{d-4}{d}b_{12}^2-\frac{3d-4}{d}b_{21}^2\\
           & \leqslant \frac{3d-4}{d^2},
    \end{align}
  \end{subequations}
  where the inequality~\eqref{eq:28} comes from Lemma~\ref{lem:normal2}.

  Hence, we can assume that $n\neq0$.  Next, maximize $h(t)$ with respect to $t$. Solve
  \begin{eqnarray}
    \frac{\mathrm{d}}{\mathrm{d} t}h(t)=\frac{n-2mt-t^2n}{(1+t^2)^2}=0,
    \label{eq:pp}
  \end{eqnarray}
  we have
  \begin{equation}
    t = \frac{-m\pm\sqrt{m^2+n^2}}{n}.
  \end{equation}

  We can further assume that $n> 0$, otherwise changing the signs of $b_{12}$ and $b_{21}$ by an orthogonal similarity
  on $B$.  Then the value of $h(t)$ at $t=\frac{-m+\sqrt{m^2+n^2}}{n}$ is greater than that at
  $t=\frac{-m-\sqrt{m^2+n^2}}{n}$. Therefore, the maximal value of $h(t)$ will be achieved at one of the following
  points:
  \begin{enumerate}[label = (\roman*)]
    \item $t=+\infty$;
    \item $t = \frac{-m+\sqrt{m^2+n^2}}{n}$.
  \end{enumerate}

  Consider case~(i): $t=+\infty$, we have
  \begin{subequations}
    \begin{align}
      h& = k= (a_1+b_{11})^2 + (a_2+b_{22}) + b_{12}^2+ b_{21}^2  \\
       & \leqslant 2(a_1^2+a_2^2+b_{11}^2+b_{22}^2+b_{12}^2+b_{22}^2)  \\
       &\leqslant 2(\norm{A}_F^2+\norm{B}_F^2) = \frac{2}{d} \leqslant \frac{3d-4}{d^2}.
    \end{align}
  \end{subequations}

  Consider the case (ii): $t=\frac{-m+\sqrt{m^2+n^2}}{ n}$.  For simplicity, we still use $h$ to indicate $\tar$.

  Note that
  \begin{equation} \frac{m+nt}{1+t^2}= \frac{m+\sqrt{m^2+n^2}}{ 2}, \end{equation} and so
  \begin{equation} \tar = h(t)= k+\frac{m+\sqrt{m^2+n^2}}{ 2}. \end{equation} The conditions $\trace{A}=\trace{B}=0$ lead to
  \begin{equation}
    a_3 = -\frac{a_1+a_2}{d-2},\quad \quad b_{33}=-\frac{b_{11}+b_{22}}{d-2}.
  \end{equation}
  In order to get rid off the linear constraints, apply the change of the variables:
  \begin{alignat}{3}
    a_1  & = x + y,   \hspace{1cm}        &a_2  &= x - y,  \\
    b_{11} & = w+z, & b_{22} &= w - z,\\\label{eq:37} b_{12} & = p + q, & b_{21} &= p - q.
  \end{alignat}
  Hence, we have
  \begin{align}
    a_3& = -\frac{a_1+a_2}{d-2}=-\frac{2x}{d-2},\\
    b_{33}&=-\frac{b_{11}+b_{22}}{d-2}
       =-\frac{2w}{d-2},
  \end{align}
  and
  \begin{align}
    \sum_{i=1}^da_i^2& = a_1^2+a_2^2+(d-2)a_3^2 
                     = (x+y)^2+(x-y)^2+\frac{4x^2}{d-2}
                     =\frac{2d}{d-2}x^2+2y^2,
  \end{align}
  \begin{subequations}
    \begin{align}
      \sum_{i,j=1}^db_{ij}^2& = b_{11}^2+b_{22}^2+(d-2)b_{33}^2 + b_{12}^2+b_{21}^2,  \\
                            & = (w+z)^2+(w-z)^2 + \frac{4w^2}{d-2}
       +(p+q)^2+(p-q)^2,  \\
                            & = \frac{2d}{d-2}w^2 + 2(z^2+p^2+q^2).
    \end{align}
  \end{subequations}
  Then $\norm{A}_F^2+\norm{B}_F^2=\frac{1}{d}$ is equivalent to
  \begin{equation}
    \frac{d}{d-2}(x^2+w^2) + y^2+z^2+p^2+q^2 =
    \frac{1}{2d}.
  \end{equation}
  Then to maximize $\tar$ is equivalent to solve the following optimization problem:
  \begin{equation}\label{eq:32}
    \begin{aligned}
      \max & \quad h(x,y,z,w,p,q)\\
      \mathrm{s.t.} &\quad \frac{d}{d-2}(x^2+w^2) + y^2+z^2+p^2+q^2 = \frac{1}{2d},
    \end{aligned}
  \end{equation}
  where
  \begin{align}
        h(x,y,z,w,p,q)
    &= \tar \\
 &   = 2\sqrt{\varDelta}+2(p^2+p q+q^2
     +z(w+x+y)+(w+x)^2+y^2+z^2),
  \end{align}
  and
  \begin{equation}
    \varDelta=  \left(p^2+z^2\right) \left(q^2+(w+x-y)^2\right).
  \end{equation}

  We can prove that $h$ achieves the maximal value only when $w=x$.  Consider the KKT condition of \eqref{eq:32}, we
  have
  \begin{align}\label{eq:33}
    \frac{\partial h}{\partial x} - \frac{2\mu d}{d-2}x &=0,\\\label{eq:35}
    \frac{\partial h}{\partial w} - \frac{2\mu d}{d-2}w &=0,
  \end{align}
  where $\mu$ is the Lagrange multiplier.  It is unlikely that $\mu=0$ in the optimization problem \eqref{eq:32}.  In
  fact, the KKT condition implies that
  \begin{equation}\label{eq:4}
    \nabla h(x,y,z,w,p,q) - 2\mu (\frac{d}{d-2}x,y,z,\frac{d}{d-2}w,p,q)^\intercal  =0.
  \end{equation}
  In particular, we have
  \begin{align}
    \frac{\partial h}{\partial x} &= 2 ( \frac{(w+x-y)(p^2+z^2)}{\sqrt{\Delta}}  + 2x + z+ 2w),\\
    \frac{\partial h}{\partial y } &= 2( -\frac{ (w+x-y)(p^2+z^2)}{\sqrt{\Delta}} + 2y + z   ),\\
    \frac{\partial h}{\partial z} & = 2( \frac{ z(q^2+( w+x-y)^2)}{\sqrt\Delta} + x +y + 2z + w),\\
    \frac{\partial h}{\partial p}& = 2 ( \frac{p(q^2+(w+x-y)^2)}{\sqrt{\Delta}}+ 2p + q ),\\
    \frac{\partial h}{\partial q} &= 2(\frac{ q(p^2+ z^2)}{\sqrt{\Delta}}+ p+ 2q  ),\\
    \frac{\partial h}{\partial  w} & = \frac{\partial h}{\partial x},\label{eq:43}
  \end{align}
  and
  \begin{equation}
    \nabla h = \left(  \frac{\partial h}{\partial x},\;  \frac{\partial h}{\partial y},\;
      \frac{\partial h}{\partial z},\;  \frac{\partial h}{\partial w}, \;
      \frac{\partial h}{\partial p},\;  \frac{\partial h}{\partial q}\right)^\intercal.
  \end{equation}
  Let
  \begin{equation}
    \kett{\varphi} =(x,y,z,w,p,q)^\intercal .
  \end{equation}
  Multiply $\kett{\psi}^\intercal$ to the l.h.s of \eqref{eq:4}, we have
  \begin{equation}\label{eq:38}
    \braa{\varphi}\nabla h= 2\mu(\frac{d}{d-2}(x^2+w^2) + y^2+z^2+p^2+q^2 )=\frac{\mu}{d}.
  \end{equation}
  The simple calculation gives that
  \begin{equation}\label{eq:39}
    \braa{\varphi}\nabla h = 2h.
  \end{equation}
  Note that $h$ represents the square sum of the largest two singular values of $X$, it is always positive. Hence, by
  ~\eqref{eq:38} and \eqref{eq:39},
  \begin{equation}\label{eq:47}
    \mu  = 2d h>0.
  \end{equation}
  Therefore, by \eqref{eq:33}, \eqref{eq:35}, \eqref{eq:43}, and \eqref{eq:47}, we have
  \begin{equation}
    x= w.
  \end{equation}
  Therefore, the optimization problem~\eqref{eq:32} is equivalent to
  \begin{equation}\label{eq:36}
    \begin{aligned}
      \max & \quad h(x,y,z,p,q)\\
      \mathrm{s.t.} & \quad \frac{2d}{d-2}x^2+y^2+z^2+p^2+q^2 = \frac{1}{2d}.
    \end{aligned}
  \end{equation}
  where
  \begin{equation}
    h = 2 \sqrt{(p^2+z^2) (q^2+(y-2 x)^2)}+2\left(p^2+p q+q^2+4 x^2+2 x z+y^2+y z+z^2\right).
  \end{equation}
  Further, if we replace $x$ by $u/\beta$ where $\beta = \sqrt{\frac{2d}{d-2}}$, then the optimization
  problem~\eqref{eq:36} is equivalent to
  \begin{equation}
    \max_{\norm{\kett v}^2=\frac{1}{2d}}  h(\kett{v}),
  \end{equation}
  where
  \begin{equation}
    h(\kett{v})= 2 \sqrt{(p^2+z^2) (q^2+(y-\frac{2 u}{\beta })^2)}
    +2(p^2+p q+q^2+\frac{4 u^2}{\beta ^2}+\frac{2 u z}{\beta }+y^2+y z+z^2).
  \end{equation}
  and
  \begin{equation}
    \kett{v} = (u,y,z,p,q).
  \end{equation}
  In the following steps, we consider this optimization problem with different cases:
  \begin{enumerate}[label = (\arabic*)]
    \item $pq=0$,
    \item  $p,q\neq 0$.
  \end{enumerate}
  Consider the case (1): $pq=0$.\\
  If $p=0$. By \eqref{eq:37}, we have
  \begin{equation}
    p = \frac{b_{12}+b_{21}}{2}=0.
  \end{equation}
  That is $b_{12}=-b_{21}$. Note that $B$ is of the shape as \eqref{eq:20}, and it is real. Hence $B$ is
  anti-symmetric, and thus normal.

  If $q=0$. By \eqref{eq:37}, we have
  \begin{equation}
    p = \frac{b_{12}-b_{21}}{2}=0.
  \end{equation}
  That is $b_{12}=b_{21}$. $B$ is then symmetric and thus normal.

  In either case, $A$ and $B$ are both normal and $f$ still represents $\tar$ despite the change of variables.  By
  Theorem~\ref{lem:normal}, it holds
  \begin{equation} \tar\leqslant \frac{3d-4}{d^2}. \end{equation}

  Next, consider the case (2): $p,q\neq 0 $. \\
  Here we apply an inequality
  \begin{equation}
    p^2+ q^2 + pq\leqslant \frac{3}{2}(p^2+q^2)
  \end{equation}
  to $h(\kett{v})$. That is
  \begin{equation}
    h(\kett{v}) \leqslant g(\kett{v}) = 2 (\sqrt\varDelta+\frac{3}{2}(p^2+q^2)+\frac{4 u^2}{\beta ^2}+\frac{2 u z}{\beta }+y^2+y z+z^2),
  \end{equation}
  where
  \begin{equation}\label{eq:48}
    \varDelta=  {(p^2+z^2) (q^2+(y-\frac{2 u}{\beta })^2)}.
  \end{equation}
  Consider the KKT condition of the optimization problem:
  \begin{equation}\label{eq:51}
    \max_{ \norm{\kett{v}}^2=\frac{1}{2d}}g(\kett{v}).
  \end{equation}
  We have
  \begin{eqnarray}
    2 \left( \frac{q^2+(y-\frac{2 u}{\beta })^2}{\sqrt\varDelta} + 3 - \mu \right)p &=0,  \\
    2\left( \frac{p^2+z^2}{\sqrt\varDelta} + 3- \mu\right)q& = 0.
  \end{eqnarray}
  Since $p,q\neq 0$, and $\varDelta>0$, the following condition is necessary for $g(\kett{v})$ reaching the maximal
  value:
  \begin{equation}\label{eq:50}
    p^2 + z^2 = q^2+(y-\frac{2 u}{\beta })^2.
  \end{equation}
  Assume \eqref{eq:50} holds. Therefore, we have
  \begin{equation}
    \sqrt{\varDelta} = \frac{1}{2}(  p^2 + z^2 + q^2+(y-\frac{2 u}{\beta })^2).
  \end{equation}
  Forward,
  \begin{equation}
    g(\kett{v}) = 4p^2 + 4q^2 +  z^2+ (y - \frac{2u}{\beta})^2 + 2( \frac{4 u^2}{\beta ^2}+\frac{2 u z}{\beta }+y^2+y z+z^2).
  \end{equation}
  Again, we can apply Lemma~\ref{lem:y:square} to simplify the optimization problem~\eqref{eq:51}.  Let
  \begin{equation}
    \kett{x }   = (u,y,z)^\intercal,\quad
    \kett{y}  = (p,q)^\intercal,\quad
    \kett{c}     = (0,0,0)^\intercal,
  \end{equation}
  and
  \begin{equation}
    f(\kett{x}) =  x_3^2+ (x_2 - \frac{2x_1}{\beta})^2 + 2( \frac{4 x_1^2}{\beta ^2}+\frac{2 x_1 x_3}{\beta }+x_2^2+x_2 x_3+x_3^2).
  \end{equation}
  Moreover, let
  \begin{equation}
    \tau_i   = \omega_i = 1,\quad
    \xi_i  = 4,\quad
    r  = \frac{1}{2d}.
  \end{equation}
  Hence, $\eta = \max\{\frac{\xi_i}{\omega_i}\}=4$.

  To maximize $g(\kett{v})$ is equivalent the solve the following optimization problem:
  \begin{equation}
    \begin{array}{rl}
      \max & \quad f(\kett{x}) +  \sum_{i=1}^2\xi_iy_i^2\\
      \text{ s.t. } & \quad
                      \begin{cases}
                        \langle \kett{c},\kett{x} \rangle & = 0,\\
                        \sum_{i=1}^3\tau_i x_i^2 +\sum_{i=1}^2\omega_i y_i^2 &= r.
                      \end{cases}

    \end{array}
  \end{equation}
  Therefore, by Lemma~\ref{lem:y:square}, the maximal value of $g$ is achieved only when $\kett{y}=0$ (i.e., $p=q=0$) or
  \begin{equation}
    \max g=\frac{4}{2d}\leqslant \frac{3d-4}{d^2}.
  \end{equation}
  For the former case, i.e., $p=q=0$, we have the result that $g= h=\tar$, where
  \begin{equation}
    \max g=  \max h \leqslant \frac{3d-4}{d^2}
  \end{equation}
  is proved in case (1).

  The latter case leads directly to
  \begin{equation}
    \max \tar = \max h \leqslant \max g \leqslant \frac{3d-4}{d^2} .
  \end{equation}
  This completes our proof.
  \end{proof}

If $\mathrm{dim}(B_1)=1$ then Lemmas \ref{lem:c1} implies that $B$ is diagonal and so it is normal. Consequently, by
Theorem \ref{lem:normal}, we have $\tar \leqslant \frac{3d-4}{d^2}$. This together with Lemmas \ref{lem:c1},
\ref{lem:c2} and \ref{lem:c3} gives the following result.

\begin{lemma}\label{lem:c4}
  Suppose $X = A\otimes\I+\I \otimes B\;(d\geqslant 4)$ with $A, B\in \bbR^{d\times d}$,
  $\trace{A}=\trace{B}=0,\break\; \norm{A}_F^2+\norm{B}_F^2 = \frac{1}{d}$, $A=\mathrm{diag}(a_1,a_2,\ldots,a_d)$ and
  $B\in \cP$.  If
  \begin{equation} \sigma_1(X) = \sigma_1(a_1\I + B_1), \quad \quad\sigma_2(X) = \sigma_1(a_2\I + B_1), \end{equation} then
  \begin{equation} \tar \leqslant \frac{3d-4}{d^2}.  \end{equation}
\end{lemma}

Now it is straightforward to have Lemma~\ref{lem:22B} by combining Lemmas~\ref{lem:c0} and \ref{lem:c4}.


\begin{IEEEbiography}[{\includegraphics[width=1in,height=1.25in,clip,keepaspectratio]{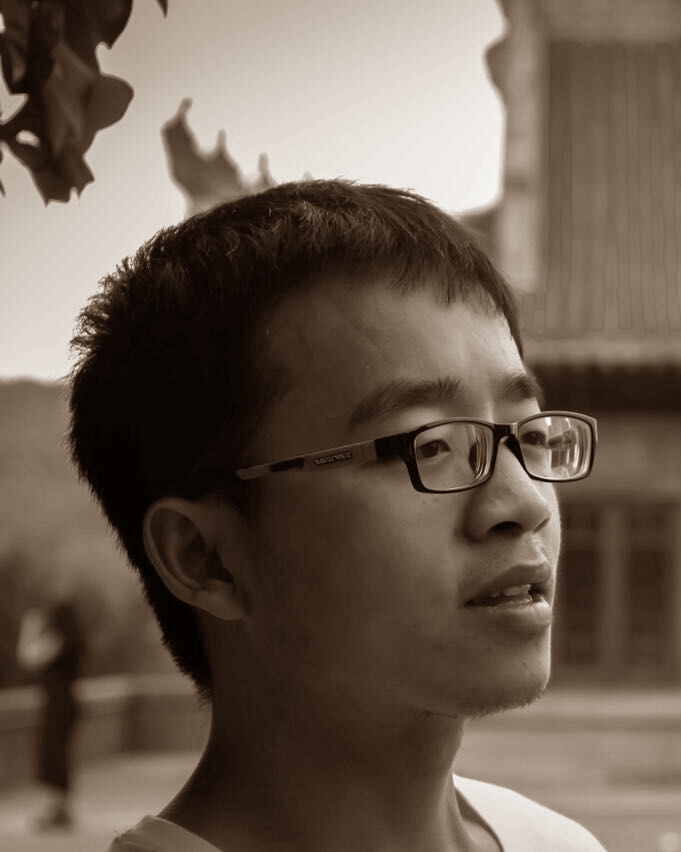}}]{Lilong Qian}
  is currently a Ph.D. candidate in the Department of Mathematics, University of Singapore.  He received the B.S.  degree in
  mathematics from Wuhan
  University, Wuhan, China,  in July of 2015.
  His research interests include the quantum information theory, numerical linear algebra,
  and optimization theory.
\end{IEEEbiography}
\begin{IEEEbiography}[{\includegraphics[width=1in,height=1.25in,clip,keepaspectratio]{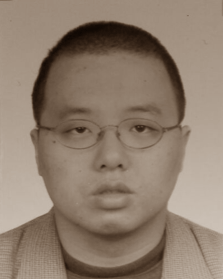}}]{Lin Chen}
  received the Ph.D. degree from the Department of Physics, Zhejiang University,
Hangzhou, China in 2008. He is currently an associate professor in the department of mathematics,
Beihang University, Beijing. His research interests include quantum information,
entanglement theory, mutually unbiased basis, tensor rank, and matrix theory.
\end{IEEEbiography}
\begin{IEEEbiography}[{\includegraphics[width=1in,height=1.25in,clip,keepaspectratio]{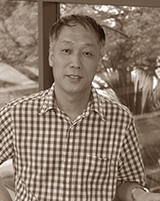}}]{Delin Chu}
received the PhD degree from the Department of Applied Mathematics, Tsinghua University, Beijing, China, in 1991. He is currently with the Department of Mathematics at National University of Singapore. He is currently an associate editor of Automatica. His research interests include data mining, numerical linear algebra, scientific computing, numerical analysis, and matrix theory and computations. 
\end{IEEEbiography}

\begin{IEEEbiography}[{\includegraphics[width=1in,height=1.25in,clip,keepaspectratio]{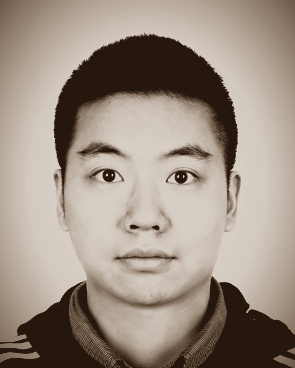}}]{Yi Shen}
is currently a Ph.D. candidate in the department of mathematics, Beihang University, Beijing, China. His research interests include the quantum information, entanglement theory, and matrix theory.
\end{IEEEbiography}

\end{document}